\documentclass[twoside,11pt]{article}

\usepackage{jair, theapa, rawfonts}

\usepackage{times} 
\usepackage{amssymb}
\usepackage{amsmath}
\usepackage{amsthm}
\usepackage[ruled, vlined,norelsize,linesnumbered]{algorithm2e}
\usepackage{graphicx}
\usepackage{regexpatch}
\usepackage{xspace}
\usepackage{nicefrac}
\usepackage{hhline}
\usepackage[round-mode=places]{siunitx}
\usepackage[black]{customlabel}
\usepackage[small]{caption}
\usepackage{subcaption}
\usepackage{wrapfig}
\usepackage{relsize}
\usepackage{bm}
\usepackage{comment}   

\usepackage{adjustbox} 

\usepackage{afterpage}  

\usepackage{tikz}
\usetikzlibrary{decorations.pathreplacing, calc, positioning}

\usepackage{color, colortbl} 

\usepackage[colorinlistoftodos,prependcaption]{todonotes}

\makeatletter
\newcounter{listing}
\newenvironment{listing}[1][htb]
  {
   \let\c@algocf\c@listing
   \begin{algorithm}[#1]%
  }{\end{algorithm}}
\makeatother

\usepackage{accents}  

\usepackage[normalem]{ulem}

\usepackage{hhline}
\usepackage{multirow}
\makeatletter
\xpatchcmd{\@todo}{\setkeys{todonotes}{#1}}{\setkeys{todonotes}{inline,#1}}{}{}
\makeatother

\newcommand{\crcell}[2][c]{
  \begin{tabular}[#1]{@{}c@{}}#2\end{tabular}}

\renewcommand{\cite}{\shortcite}
\newcommand{\citet}{\shortciteA}

\DeclareMathOperator*{\argmax}{arg\,max}
\DeclareMathOperator*{\Exp}{\mathbb{E}}
\DeclareMathOperator*{\bigtimes}{\times}
\newcommand{\eu}{\bar{u}}
\newcommand{\bru}{\eu^{BR}}
\newcommand{\loss}{l}
\newcommand{\Unif}{\mathcal{U}}
\renewcommand{\epsilon}{\varepsilon}
\newcommand*\diff{\mathop{}\!\mathrm{d}}
\newtheorem{remark}{Remark}
\newtheorem{theorem}{Theorem}
\newtheorem{lemma}{Lemma}
\newtheorem{assumption}{Assumption}
\newtheorem{definition}{Definition}

\newcommand{\sm}{\textrm{-}}
\newcommand{\smi}{{\sm i}}
\newcommand{\LLG}{{LLG}\xspace}
\newcommand{\MMLLLLGG}{{multi-minded LLLLGG}\xspace}
\newcommand{\numSettings}                  {\customlabeloutput{16}}
\newcommand{\numRuns}                       {\customlabeloutput{50}}
\newcommand{\numSettingsTimesRuns}          {\customlabeloutput{800}}
\newcommand{\targetEpsilon}                {\customlabeloutput{0.00001}}

\newcommand{\dampeningFactorNaive}         {\customlabeloutput{0.5}}
\newcommand{\maxSamplesNaive}              {\customlabeloutput{200,000}}

\newcommand{\numMCSamplesQuasi}            {\customlabeloutput{200,000}}
\newcommand{\numMCSamplesCommon}            {\customlabeloutput{10,000}}
\newcommand{\pctMCSamplesCommonToBaseline} {\customlabeloutput{5\%}}
\newcommand{\factorMCSamplesCommonToBaseline} {\customlabeloutput{20}}

\newcommand{\numCtrlPtsNotAdaptive}        {\customlabeloutput{160}}
\newcommand{\numCtrlPtsAdaptive}        {\customlabeloutput{40}}

\newcommand{\numCtrlPtsNaiveStopping}      {\customlabeloutput{1,000}}

\newcommand{\numSamplesMSDVerif}           {\customlabeloutput{200,000}}

\begin{document}


\jairheading{1}{1993}{1-15}{6/91}{9/91}
\thispagestyle{empty} 

\newcounter{tempfootnote}
\setcounter{tempfootnote}{\value{footnote}}
\setcounter{footnote}{0}
\renewcommand*{\thefootnote}{\fnsymbol{footnote}}

\title{Computing Bayes-Nash Equilibria in \\ Combinatorial Auctions with Verification\footnotemark}

\footnotetext[1]{This paper is a significantly extended version of \citet{bosshard2017fastBNE}, which was published in the conference proceedings of IJCAI'17.}

\setcounter{footnote}{\value{tempfootnote}}
\renewcommand{\thefootnote}{\arabic{footnote}}

\author{\name Vitor Bosshard \email bosshard@ifi.uzh.ch \\
       \addr Department of Informatics, University of Zurich\\
       \AND
       \name Benedikt B\"unz \email buenz@cs.stanford.edu \\
       \addr Department of Computer Science, Stanford University\\
       \AND
       \name Benjamin Lubin \email blubin@bu.edu \\
       \addr Questrom School of Business, Boston University\\
       \AND
       \name Sven Seuken \email seuken@ifi.uzh.ch \\
       \addr Department of Informatics, University of Zurich}

\maketitle





\begin{abstract}
\noindent
We present a new algorithm for computing pure-strategy $\epsilon$-Bayes-Nash equilibria ($\epsilon$-BNEs) in combinatorial auctions with continuous value and action spaces. An essential innovation of our algorithm is to separate the algorithm's \emph{search phase} (for finding the $\epsilon$-BNE) from the \emph{verification phase} (for computing the $\epsilon$). Using this approach, we obtain an algorithm that is both very fast and provides theoretical guarantees on the $\epsilon$ it finds. Our main technical contribution is a verification method which allows us to upper bound the $\epsilon$ across the whole continuous value space without making assumptions about the mechanism. Using our algorithm, we can now compute $\epsilon$-BNEs in multi-minded domains that are significantly more complex than what was previously possible to solve. We release our code under an open-source license to enable researchers to perform algorithmic analyses of auctions, to enable bidders to analyze different strategies, and to facilitate many other applications.

\end{abstract}
\section{Introduction}

Combinatorial auctions (CA) are used to allocate multiple, indivisible goods to multiple bidders.
CAs allow bidders to express complex preferences on the space of all bundles of goods, taking into account that goods can be complements or substitutes \cite{CramtonEtAl2006CombAuctions}.
They have found widespread use in practice, including for the sale of radio spectrum licenses \cite{Cramton2013SpectrumAuctionDesign}, for the procurement of industrial goods \cite{Sandholm2013largescale}, and for the allocation of TV ad slots \cite{goetzendorff2014core}.

Unfortunately, the strategyproof VCG mechanism \cite{vickrey1961counterspeculation,clarke1971multipart,groves1973incentives} has several serious flaws when applied to the CA setting: most notably, it can lead to very low or
even zero revenues despite high competition for the goods \cite{ausubel2006lovely}.
Furthermore, it leaves incentives for collusion \cite{DayMilgrom2008CoreSelectPackageAuctions}.
For these reasons, many CAs conducted in practice do not use VCG, instead opting for alternative payment rules, such as first-price payments (also known as ``pay-as-bid'') or other payment rules, typically from the class of core-selecting rules, which are designed to make the winners' payments high enough to guarantee envy-freeness \cite{DayMilgrom2008CoreSelectPackageAuctions}.
These alternative mechanisms are not strategyproof in general settings, and the behavior of bidders under them is not well understood.
If we want to predict auction outcomes in terms of desirable properties (such as incentives, revenue, efficiency), we must therefore study them \emph{in equilibrium} instead of \emph{at truth}.
As a first step, this requires us to choose a suitable equilibrium concept.
\subsection{Equilibria in CAs}

In the full information setting, significant theoretical work has been done towards characterizing the equilibria of CAs.
For first-price payments, the set of Nash equilibria has been fully characterized (\cite{bernheim1986menu}; see \cite[Chapter 8.2]{milgrom2004putting} for a modern treatment). \citet{DayRaghavan2007FairPayments} showed that an analogous characterization holds for any payment rule that always selects bidder-Pareto-optimal core payments.
While the full information Nash equilibrium (NE) may be a good approximation of bidder behavior in some settings (e.g., repeated auctions where bidders can reasonably be assumed to know others' values), it has several issues:
there is a high multiplicity of equilibria, and each equilibrium must be supported by very precise bids on losing packages.
The latter issue is especially problematic, since in a full information setting, a losing bidder would know a priori that he will not win any items, and thus has no reason to participate in the auction at all.
However, this bidder's losing bid may be needed to keep the winners in equilibrium.
For more details on the issues of equilibrium selection and stability, see the discussions in \citet{bernheim1986menu} and \citet{DayRaghavan2007FairPayments}.

In addition to these theoretical issues, many real-world high-stakes auctions such as spectrum auctions are only conducted once, and bidders work hard to keep their private information secret, motivating the study of \emph{incomplete information} settings.
In such a setting, private information is explicitly modelled by assuming that each bidder knows his own valuation but only has a prior belief (i.e., a distribution) over the valuations of others.
This leads to the solution concept of the \emph{Bayes-Nash Equilibrium (BNE)}, where bidders maximize their expected utility over many possible auction outcomes, weighted according to their beliefs.

Some analytical research into BNEs already exists.
Non-combinatorial single-item auctions have been studied extensively  \cite{klemperer1999auction}, but comparatively little is known about multi-item auctions, as the difficulty of finding BNEs by hand increases markedly, requiring the solution of challenging differential equations.
For this reason, only a few analytical results exist in small settings, most notably the Local-Local-Global (LLG) domain with two goods and three bidders (which we define in Section~\ref{sec:llg}).
\citet{Goeree2013OnTheImpossibilityOfCoreSelectingAuctions}
as well as \citet{ausubel2020core} have
independently derived the analytical BNE of the VCG-nearest or ``quadratic'' rule, which is commonly used in practice
\cite{DayCramton2012Quadratic}.
Furthermore, \citet{AusubelBaranov2013CoreOldVersion} have also derived analytical BNEs of three other payment rules.
For first-price payments, \citet{baranov2010exposure} provides some necessary properties of BNEs, but does not fully characterize them.

As known analytical methods are not amenable to larger settings, the only feasible approach to finding BNEs in such settings  is through algorithmic methods.  However, one challenge we face when designing algorithms for finding BNEs in CAs is that equilibria are not known to exist in general, but only in specific settings.\footnote{The general CA problem has a structure much richer than what is covered by the single crossing condition \cite{Athey2001SingleCrossingPureStrategy} or other approaches for games with discontinuous payoffs (see, e.g., \citet{mclennan2011games}). There have been some results for restricted classes of auctions, such as uniform price auctions \cite{mcadams2006monotone}, double auctions \cite{jackson2005existence}, and games with discrete actions and continuous single-dimensional types \cite{Rabinovich2013ComputingBNEs}.}
Fortunately, every strategy profile is an $\epsilon$-BNE for some $\epsilon$.
Thus, finding an $\epsilon$-BNE with $\epsilon$ being as small as possible is a well-defined problem, and an algorithm for solving this problem can be an important tool for auction designers to analyze existing mechanisms and to design new ones.
%

\subsection{Prior Algorithmic Work on Computing BNEs}

Computer scientists have long worked on algorithms for computing
equilibria in non-cooperative games. The \emph{Gambit} software
package provides a number of algorithms to find NEs and BNEs
\cite{McKelvey1996ComputationofEquilibria,McKelvey2016Gambit},
but only for finite games (with finite type and action
spaces). Solving auction games with even a modest number of types (valuations) and actions quickly becomes infeasible with these general solvers; therefore, infinite games can only be modeled with significant loss of fidelity.
This is why researchers have turned towards developing special-purpose algorithms for computing BNEs in CAs.
One important class of BNE algorithms is based on \emph{iterated best
  response}. The algorithms proposed by
\citet{reeves2004computing}, \citet{vorobeychik2008stochastic} and
\citet{Rabinovich2013ComputingBNEs} belong to this class.
To keep the computation manageable, all three
algorithms \emph{simplify} the strategy space (using piecewise linear
strategies, multiplicative shading strategies, or a finite
set of actions, respectively) before \emph{solving} the simplified auction game.

A fundamental limitation of this ``simplify, then solve'' approach is that the $\varepsilon$-BNE computed is only valid within the space over which the algorithm searches for best responses.
If this issue is not handled carefully, it can lead to what we call the \emph{false precision problem}: the modelling choices that were made to speed up the computation end up distorting the auction game in meaningful ways, and the equilibrium that is calculated might not be as good as the algorithm reports.\footnote{Note that \citet{vorobeychik2008stochastic}[Sec. 7.4] correctly state this  limitation of their algorithm. \citet{Rabinovich2013ComputingBNEs}
  also handle this issue correctly by only claiming to find the BNE in
  the ``game with the restricted strategy space.''. \citet{reeves2004computing} restrict themselves to a class of auctions where the best response is guaranteed to lie in the subset of the strategy space they search through.}
To illustrate this point, consider the following simple but striking thought experiment: We search for an $\epsilon$-BNE in a
CA, but restrict the bidders' actions to bidding zero on all bundles of items. Any iterated best response algorithm
will immediately find an $\varepsilon$-BNE with $\varepsilon=0$,
as there is no beneficial deviation. Obviously, this $0$-BNE only
``survives'' because bidders are artificially prevented from meaningfully participating in the auction.

Restricting the action space of a game in order to make equilibrium computation easier is known as action abstraction \cite{sandholm2015abstraction}.
For some types of finite games, abstraction methods have been developed that guarantee that an equilibrium of the abstracted game can be translated into an equilibrium of the original game, with a bound on how much the coarseness of the abstraction affects the solution quality \cite{sandholm2012lossy}.
Unfortunately, no such methods exist for infinite games.
A recent algorithm proposed by \citet{bosshard2018nondecreasing} could be interpreted as a kind of action abstraction for CAs, but it is only applicable to CAs where the payment rule is non-decreasing (e.g., first-price auctions).

\subsection{Overview of our Contribution}

In this paper, we develop a fast algorithm for computing pure-strategy $\epsilon$-BNEs in CAs, without sacrificing accuracy in the formalization of the auction game or the equilibrium computed within this game.
For this purpose, we set up an algorithm framework  (Section~\ref{sec:bneAlgorithmFramework}) that protects against the false precision problem.
This framework splits the task of computing an $\epsilon$-BNE into a a \emph{search phase} and a \emph{verification phase}.
%

In the search phase, our goal is to find an $\epsilon$-BNE as quickly as possible.
For this, we design an algorithm based on iterated best response that is highly optimized to CAs  (Section~\ref{sec:compBestResponses}).
We use many numerical techniques to cut down on computation time.
The result of the search phase is a BNE  candidate: a strategy profile likely to offer only small incentives to deviate for any bidder, but for which the algorithm, at this time, only has an estimate of the $\epsilon$ (i.e., the incentives to deviate).

In the verification phase, our goal is to find an $\epsilon$ such that the equilibrium candidate is in fact an $\epsilon$-BNE.
For this, it is important to compare the quality of the equilibrium candidate against the best alternative strategies available in the \textit{full} strategy space, without the restrictions imposed during search.
The main technical contribution of our paper is a verification procedure that computes a theoretical upper bound on $\epsilon$, taking into account the full strategy space of the auction game, i.e., the entire continuum of valuations and actions (Section~\ref{sec:provenErrorBound}). It is surprising that it is even possible to derive such a theoretical bound, given that any algorithm can only evaluate a finite number of individual valuations and actions, while the bound requires reasoning over the continuous value and action spaces.
To the best of our knowledge, our approach is the first to achieve such guarantees for infinite games without restricting the strategy space.
Our results require some mild theoretical conditions to hold (most notably, independent distributions of bidders' valuations), so we also provide a second verification procedure that robustly estimates $\epsilon$ in arbitrary CA settings.
This alternative approach ensures that our algorithm is fully general. While it only produces an estimated $\epsilon$, we also show experimentally that, with sufficient computation time, the estimated $\epsilon$ and the upper bound on $\epsilon$ converge towards each other.

We validate our approach by running a series of experiments in the \LLG domain with two items and three bidders.
Our experiments show that the techniques we use significantly speed up our algorithm's runtime and that our algorithm converges consistently despite the use of randomness. Moving beyond single-dimensional domains like \LLG, we also discuss the difficulties of scaling any BNE algorithm to high-dimensional auctions (Section~\ref{sec:HigherDims}).
We introduce the new \MMLLLLGG domain with eight goods and six bidders, which is much larger than any domain that previous algorithms have been able to tackle.
In this domain, we find accurate $\epsilon$-BNEs for both the VCG-nearest and first-price payment rules, which demonstrates the scaling capabilities of our algorithm, and sets a benchmark for future work on BNE algorithms.

There are multiple use cases for our BNE algorithm. First, researchers may use it for the purpose of auction design, e.g., when automatically searching for optimal payment rules within a given design space \cite{Lubin2015AMMAAbstract,lubin2018designing}. Other researchers may find our algorithm useful when analyzing specific aspects of CAs, like the impact of reserve prices \cite{DayCramton2012Quadratic}, incentives for overbidding in CAs \cite{beck2013incentives}, or as a supporting tool for finding and validating new analytical results \cite{baranov2010exposure}. Actual bidders participating in an auction could use our algorithm to optimize their bidding, e.g., by analyzing the effect of different strategies under a given CA design. Further, bidders that are contemplating whether to participate in an auction could use our algorithm to better understand their economic position in a given competitive landscape. To enable all of these use cases, we explain how to use our software  and describe its main features (Section~\ref{sec:software}), and we release our code under
an open-source license at \texttt{https://github.com/marketdesignresearch/CA-BNE}.

\section{Preliminaries}
\label{sec:formal_model}
\subsection{Formal Model}
\subsubsection*{Combinatorial Auctions}

A combinatorial auction (CA)  is used to sell a set $M = \{1, 2, \ldots, m \}$ of goods to a set $N = \{1, 2, \ldots, n \}$ of bidders.
Each bidder $i$ has a  value $v_i(K) \in \mathbb{R}_{\geq 0}$ for each bundle of goods $K \in 2^M$. We assume that these values are normalized such that $v_i(\emptyset)= 0$.
It is often the case that a bidder is only interested in a small set of $d$ different bundles $\{K_1, \ldots, K_d\}$, having strictly positive value for each of those $d$ bundles, and value $0$ for all others. Therefore, it is convenient to represent the bidder's \emph{valuation} $v_i$ as a point in the $d$-dimensional \emph{value space} $\mathbb{V}_i \subseteq \mathbb{R}_{\geq 0}^{d}$, where $d$ can be between $1$ and $2^m-1$, depending on the setting.


The bidder submits a (possibly non-truthful) bid to the auction. To simplify the exposition, we adopt the XOR bidding language \cite{Nisan2006BiddingLanguagesForCombinatorialAuctions}, but our algorithm and results generalize to other bidding languages. In our model, every bidder is allowed to submit an XOR bid with exactly $r\in \mathbb{N}$ atomic bids (i.e., expressing a value for $r$ different bundles), where the $r$ and the set of atoms is fixed a priori for each bidder.
Note that setting $r = 2^m - 1$ allows bidders to express arbitrary valuations over the entire bundle space. This modeling choice is thus without loss of generality. However, our model also allows for smaller $r$, capturing bidders who focus their attention on a specific set of bundles. Given this, bidder $i$'s bid can be represented by a point $b_i$ in the \emph{action space} $\mathbb{R}_{\geq 0}^{r}$.
We do not assume free disposal, and therefore the value for all bundles for which a bidder does not submit a bid is implicitly zero.
The bid profile $b = (b_1, \ldots, b_n)$ is the vector of all bids.
The bid profile of every bidder except $i$ is denoted $b_{\smi}$.

The CA has an allocation rule $X$ assigning bundle $X_i(b)$ to bidder $i$, in such a way that the allocation it produces is \emph{feasible},  i.e., $\forall i,j \in N: X_i(b) \cap X_j(b) = \emptyset$.\footnote{In much of the literature, it is typically assumed that the allocation rule is also \emph{efficient} (i.e., maximizes the sum of bidders' reported values). Of course, our model covers this case, but it is much more general, as it can handle \emph{arbitrary} allocation rules, including approximately efficient ones, or rules targeting other objectives such as revenue.} The CA also has a payment rule $p$ which is a function assigning a payment $p_i(b) \in \mathbb{R}_{\geq 0}$ to each bidder.
We let 
$u_i(v_i, X(b), p(b))$ denote bidder $i$'s utility for an auction outcome, given his own valuation $v_i$, an allocation $X(b)$, and payments $p(b)$.\footnote{To simplify the language, we always use he/his when referring to a bidder.}
Since the allocation and payment rules $X$ and $p$ are always fixed in each auction, bidder $i$'s utility only depends on the bid profile $b$, so we abbreviate the utility as $u_i(v_i, b)$, or equivalently $u_i(v_i, b_i, b_\smi)$.

\subsubsection*{CAs as Bayesian Games}
\label{sec:bayesiangames}

We model the process of bidding in a CA as a Bayesian game.\footnote{In the game theory literature, a bidder would be called a \emph{player}, and his valuation $v_i$ would be called his \emph{type}.}
Each bidder knows his own valuation $v_i$, but he only has probabilistic information (i.e., a prior) over each other bidder $j$'s valuation $v_j$. {\color{black}
This information is represented by the random variable $V_j$ drawn from the distribution $\mathcal{V}_j$ with support $\mathbb{V}_j$ (i.e., all valuations in $\mathbb{V}_j$ have strictly positive probability under $\mathcal{V}_j$). The joint prior $\mathcal{V} = (\mathcal{V}_1, \ldots, \mathcal{V}_n)$ is common knowledge and consistent among all bidders.
}

{\color{black}
The goal of bidder $i$ is to maximize his utility $u_i(v_i, b_i, b_\smi)$ {in expectation} over the distribution of bids $b_\smi$ of all other bidders.
To capture a bidder's belief about $b_\smi$, we model every bidder $j$ as having a strategy $s_j: \mathbb{R}_{\geq 0}^d \mapsto \mathbb{R}_{\geq 0}^r$, which is a function mapping all of his possible valuations to bids.\footnote{\color{black} With this definition, we assume that all bidders' strategies are pure and not mixed. This assumption is quite natural from a behavioral perspective. Furthermore, pure-strategy $\epsilon$-BNEs (introduced further below) always exist for some $\epsilon >0$. This, together with our verification method (Section~\ref{sec:provenErrorBound}), ensures that the $\epsilon$-BNE-finding problem is well-defined. Nevertheless, allowing for mixed strategies would be an interesting avenue for future work, as it models different economic behavior and could even improve the accuracy and convergence of our algorithm.}
From the point of view of bidder $i$, the bids $b_\smi$ are thus induced by the valuations $v_\smi$ through the strategies: $b_\smi := s_\smi(v_\smi)$.
This allows us to introduce the \emph{expected utility} $\eu_i$, defined as
}
%
%
%

%
%
\begin{equation}
    \eu_i(v_i, b_i) := \Exp_{V_{\smi}\sim \mathcal{V}_{\smi} | V_i = v_i} \left[ u_i(v_i, b_i, s_\smi(V_\smi)) \right],
    \label{eq:util}
\end{equation}
where $\mathcal{V}_{\smi} | V_i = v_i$ is the conditional distribution of valuations of all other bidders, which can depend on the realization of $V_i$ (when the distributions of valuations are not independent).

The highest possible expected utility that can be achieved with any bid is the \emph{best response utility}, given by
\begin{equation}
    \bru_i(v_i) = \sup_{b_i' \in \mathbb{R}_{\geq 0}^{r}} \eu_i(v_i, b_i').
    \label{eq:brutil}
\end{equation}
Note that we take the supremum over bids instead of the maximum, because the maximum might not exist due to discontinuities in the utility function.\footnote{Consider, e.g., a single-item first-price auction with complete information. If opponents bid a maximum amount of $x$, the best response is often to outbid them by a small amount, i.e., to bid $\lim_{\delta \to 0} x + \delta$. Analogues of this situation can arise even in incomplete information settings, where a discontinuity arises due to $i$'s bid crossing over a threshold where his probability of winning a certain bundle jumps by a discrete amount. Such thresholds are caused by point masses in the distribution of $s_\smi(V_\smi)$, which can occur even if the distribution of $V_\smi$ itself is smooth, e.g., when strategies have flat segments.}
Whenever a bidder $i$ submits a bid $b_i$ that is not optimal he leaves a certain amount of utility ``on the table.''
We call this quantity the \emph{utility loss}, given by
\begin{equation}
    \loss_i(v_i, b_i) := \bru_i(v_i) - \eu_i(v_i, b_i).
    \label{eq:loss}
\end{equation}

Bidders are in an $\epsilon$-Bayes-Nash equilibrium when the utility loss is smaller than $\epsilon$ for all possible valuations of all bidders, i.e., no bidder has a profitable deviation from the equilibrium netting him more than $\epsilon$ utility in expectation:

\begin{definition}
    An \textbf{$\epsilon$-Bayes-Nash equilibrium} ($\epsilon$-BNE) is a strategy profile $s^*$ such that
    $$\forall i \in N, \forall v_i \in \mathbb{V}_i: \loss_i(v_i, s_i^*(v_i)) \leq \epsilon.$$
\end{definition}
We take the $\epsilon$-BNE as our solution concept because we use numerical algorithms with limited precision to find the BNEs. Thus, when we \emph{solve} a CA, we mean that we find an $\epsilon$-BNE, where $\epsilon$ is a suitably small constant fixed a priori.

\begin{remark}
Throughout the paper, we sometimes refer to the ``true'' $\epsilon$ of a strategy profile $s^*$. By this we mean the \emph{smallest} $\epsilon$ such that $s^*$ is an $\epsilon$-BNE.
\end{remark}
\subsection{The \LLG Domain and Straightforward Bundle Bidding}
\label{sec:llg}

In this paper, we study the performance of our BNE algorithm, first in a small domain,
where analytical results are available, and later in a novel larger
domain (Section~\ref{sec:HigherDims}).
For the former, we use the widely-studied Local-Local-Global (\LLG)
domain \cite{ausubel2006lovely}.
\LLG is one of the smallest examples of an auction where combinatorial interactions between bidders arise.
There are three bidders, with bidders 1 and 2 being local, interested in two different single goods, and bidder 3 being global, interested in the package of both goods.

This domain has been widely studied, and analytical BNE results for it are available under several mechanisms.
\citet{AusubelBaranov2013CoreOldVersion}
study the case where the global bidder's valuation is drawn from $\Unif[0,2]$, while the local bidders' valuations are drawn from $\Unif[0,1]$, with cumulative distribution function $F(v) = v^\alpha$.
Furthermore, the local bidders' valuations are perfectly correlated with probability $\gamma$, and independent otherwise.
{\color{black}
Within this framing, they derive analytical BNEs for four different mechanisms. All of the mechanisms use the efficient allocation rule but each of them uses a different core-selecting payment rule (VCG-nearest, Nearest-bid, Proxy and Proportional).}
These BNEs are unique under the assumption that local bidder strategies are symmetric.
Adopting their
results as our benchmark, we assemble a set of \numSettings \textit{auction
settings} to be used as a test suite: four payment rules each applied to four domains ($\alpha \in
\{1,2\} \times \gamma \in \{0, 0.5\}$).
Note that our selection of settings covers all payment rules for which analytical results are known and all interesting regions in the $(\alpha, \gamma)$-parameter space for which the analytical results have distinct functional forms \cite{AusubelBaranov2013CoreOldVersion}.\footnote{We excluded some simple settings where the analytical solution is just a multiplicative shade.} Importantly, we did not exclude any settings for modeling or algorithmic tractability reasons.

In the literature on \LLG, it is often assumed (sometimes implicitly) that the local bidders can only bid on the good they are directly interested in \cite{Goeree2013OnTheImpossibilityOfCoreSelectingAuctions,ausubel2020core}.\footnote{\color{black}There are two noteworthy exceptions. \citet{beck2013incentives}  analyzed a specific payment rule in LLG with $r=3$, thus allowing bids on all non-empty bundles in this case. More recently, \citet{Bosshard2020Cost}  studied bidders' incentives for exploiting an exponentially large action space, requiring $r=2^m$.}
This implements what we call ``straightforward bundle bidding:'' The bidder submits an atomic bid on each bundle that has strictly positive \emph{marginal value}, i.e., where removing any part of the bundle would strictly decrease the value.
Formally:
\begin{definition}
    An XOR bid $b_i$ is a \emph{straightforward bundle bid} if the bid contains one atomic bid for each bundle in the set
    \begin{equation}
        \{ K \subseteq M \enskip | \enskip \forall K' \subset K : v_i(K') < v_i(K) \}
        \label{eq:straightforward}
    \end{equation}
    and no other atomic bids.
\end{definition}
{\color{black}
The interpretation of the set (\ref{eq:straightforward}) is simple: To declare the valuation $v_i$ truthfully, the bid $b_i$ must include an atomic bid for each of the subsets in (\ref{eq:straightforward}), and no other atomic bids are necessary.
A straightforward bundle bid is thus the smallest bid capable of expressing a bidder's true valuation.}
For our analysis of LLG, we follow the majority of prior work and also assume straightforward bundle bidding. Because in LLG, every bidder is single-minded (i.e., only interested in one bundle), we thus set  $r=1$. Given that all of our payment rules are minimum-revenue core-selecting, we know that it is a dominant strategy for the global bidder to be truthful \cite{beck2013incentives}. Thus, in LLG, computing the $\epsilon$-BNE only involves computing the strategies for the local bidders. Furthermore, to match the analytical results of
\citet{AusubelBaranov2013CoreOldVersion},
we only consider symmetric equilibria in this domain (though this simplification is not
essential to our algorithm).
Thus, the whole strategy profile in LLG can be described by a single function  $s_{local}: [0,1] \mapsto \mathbb{R}_{\geq 0}$.
\section{BNE Algorithm Framework}
\label{sec:bneAlgorithmFramework}

\begin{algorithm}[tb]
\SetAlgoLined
\DontPrintSemicolon
\SetKwInOut{KwInput}{input}
\SetKwInOut{KwOutput}{output}
\SetKwInOut{KwParameters}{parameters}

\KwInput{Mechanism $\mathcal{M}$, Distribution $\mathcal{V}$ of bidder valuations}
\KwOutput{$\varepsilon$-BNE strategy profile}
\SetKwFunction{convergedInner}{ConvergedInner}
\SetKwFunction{convergedOuter}{ConvergedOuter}
\SetKwFunction{bestResponseInner}{BestResponseInner}
\SetKwFunction{bestResponseOuter}{BestResponseOuter}
\SetKwFunction{update}{Update}
\SetKwFunction{utilityLoss}{UtilityLoss}
\SetKwFunction{verification}{Verification}
\SetKwFunction{convertStrategies}{ConvertStrategies}
\tikz[overlay,remember picture]{\coordinate (mark1)}
$s := $ truthful strategies \;
\tikz[overlay,remember picture]{\coordinate (mark2)}
\Repeat{$\convergedOuter(\tilde{\epsilon})$ \tikz[overlay,remember picture]{\coordinate (mark4)}}{
    \Repeat{$\convergedInner(\tilde{\varepsilon})$\tikz[overlay,remember picture]{\coordinate (mark3)}}{
        \ForEach{\upshape bidder $i$}{
            $s_i' := \bestResponseInner(\mathcal{M},\mathcal{V},s_\smi)$ \label{algLine:BR}\;
        }
        $\tilde{\varepsilon} := \utilityLoss(s, s')$ \;
        $s := \update(s, s')$ \;
    }
    \ForEach{\upshape bidder $i$}{
        $s_i' := \bestResponseOuter(\mathcal{M},\mathcal{V},s_\smi)$ \;
    }
    $\tilde{\varepsilon} := \utilityLoss(s, s')$ \;
    $s := \update(s, s')$ \;
}
$s^* := \convertStrategies(s)$ \label{algLine:Convert}\;
\tikz[overlay,remember picture]{\coordinate (mark5)}
$\varepsilon := \verification(s^*)$ \label{algLine:Verification}\;
\Return{$(s^*, \varepsilon)$}%
\begin{tikzpicture}[overlay,remember picture]


    \coordinate (inner1)  at (mark2 -| 7.2cm,0);
    \coordinate (inner2)  at (mark3 -| 7.2cm,0);
    \coordinate (outer1)  at (mark1 -| 8.5cm,0);
    \coordinate (outer2)  at ({$(mark4)+(0,0.05cm)$} -| 8.5cm,0);
    \coordinate (search1) at (outer1 -| 9.8cm,0);
    \coordinate (search2) at (outer2 -| 9.8cm,0);
    \coordinate (verif1)  at (mark4 -| 9.8cm,0);
    \coordinate (verif2)  at ({$(mark5)-(0,0.1cm)$} -| 9.8cm,0);

    \draw [decorate,decoration={brace,amplitude=7pt,raise=4pt}]
    (outer1) -- (outer2) node [black,midway,xshift=0.6cm, text width=0.cm] {\small outer loop};

    \draw [decorate,decoration={brace,amplitude=7pt,raise=4pt}]
    (inner1) -- (inner2) node [black,midway,xshift=0.6cm, text width=0.cm] {\small inner loop};

    \draw [decorate,decoration={brace,amplitude=7pt,raise=4pt}]
    (search1) -- (search2) node [black,midway,xshift=0.6cm, text width=0.cm] {\small search phase};

    \draw [decorate,decoration={brace,amplitude=7pt,raise=4pt}]
    (verif1) -- (verif2) node [black,midway,xshift=0.6cm, text width=0.cm] {\small verification phase};

\end{tikzpicture}
\caption{Iterated Best Response with Verification}
\label{alg:IBR}
\end{algorithm}

In this section, we present a high-level overview of our BNE algorithm.
{\color{black}
At its core, the algorithm is based on \emph{iterated best response}, which simulates the process of each bidder repeatedly updating his strategy to  be a best response to the other bidders' previous strategies, terminating when the utility loss across all bidders is small enough.
As such, our algorithm belongs to a broad class of algorithms based on adaptive dynamics \cite{milgrom1990rationalizability}, which also includes the well-known fictitious play algorithm \cite{Brown1951iterative}. See \citet{reeves2004computing} for a brief historical overview.
}

The key difference between our algorithm and the standard iterated best response procedure is that our algorithm is separated into a \emph{search phase} and a \emph{verification phase} (with the search phase being further split into an inner loop and an outer loop).
Our full algorithm is presented as Algorithm~\ref{alg:IBR} and described below at a high level.

In the first step, we initialize the strategy profile to be used during the search at truth (Line 1).\footnote{\color{black} 
The algorithm can alternatively be initialized at any other strategy profile, which could lead to finding different equilibria. However, starting at the truthful strategy profile is quite natural, and we expect it to lead to equilibria more likely to be encountered in practice.
With the proper experiment setup and given enough computational resources, our algorithm can be used to investigate questions regarding equilibrium multiplicity and robustness as well.
}
The search phase constitutes the first part of the algorithm (Lines 2-15). The core of the search phase consists of the \emph{inner loop} (%
Lines 3-9), in which iterated best response is performed.
Each iteration consists of calculating a best response for each bidder in turn (Lines 4-6), computing the utility loss (Line 7), updating the current strategy profile (Line 8), and checking for convergence (Line 9).
Note that the strategy update  we perform is a \emph{dampened update}  such that each bidder's new strategy is a mixture of his old strategy and his best response.\footnote{This is a standard method to prevent oscillations around the solution, a phenomenon typical of any procedure that iteratively searches for fixed points.}
When the inner loop converges,  we repeat the same steps in the outer loop (Lines 10-15), this time configured with higher precision.
If the outer loop fails to converge, we go back into the inner loop again, and so on.
These two nested loops help the algorithm to converge reliably without expending too much computational effort; we will discuss the role of the outer loop in more detail in Section~\ref{sec:puttingitalltogether}.
Once the outer loop has converged as well (Line 15), we have arrived at a strategy profile that is an equilibrium candidate, which concludes the search phase.
We then proceed to the verification phase (Lines 16-17), where the $\epsilon$ corresponding to this equilibrium candidate is computed.
Finally, we return a strategy profile $s^*$ and $\epsilon$ (Line 18), such that $s^*$ is an $\epsilon$-BNE.

To instantiate Algorithm~\ref{alg:IBR}, we still need to define how to compute best responses (Line \ref{algLine:BR}) and how to perform verification (Line \ref{algLine:Verification}). Turning to the former, the best response $BR_i$ is a strategy for bidder $i$ maximizing $i$'s expected utility  $\eu_i(v_i, b_i )$ for each possible valuation $v_i$:
\begin{equation}
  \label{eq:brExp}
  BR_i(v_i) := \argmax_{b_i \in \mathbb{R}_{\geq 0}^{r}} \, \eu_i(v_i, b_i).
\end{equation}
Here, the expected utility $\eu_i(v_i, b_i )$ is calculated with respect to the strategy profile $s_\smi$ of the previous round.\footnote{Note that a best response need not always exist due to discontinuities in the utility function (as discussed in Section~\ref{sec:bayesiangames}). Because of this, in our algorithm, we must accept a bid $b_i$ with utility very close to the best response utility $\bru_i(v_i)$ in place of a true best response.} 
In the rest of this section, we go into the details of how to compute the best response according to Equation~\eqref{eq:brExp}, and we provide an overview of the verification phase.

\subsection{Modeling Strategies}
\label{sec:modelingStrategies}

\begin{figure}
\centering

\begin{subfigure}{0.49\textwidth}
\includegraphics[width=\textwidth]{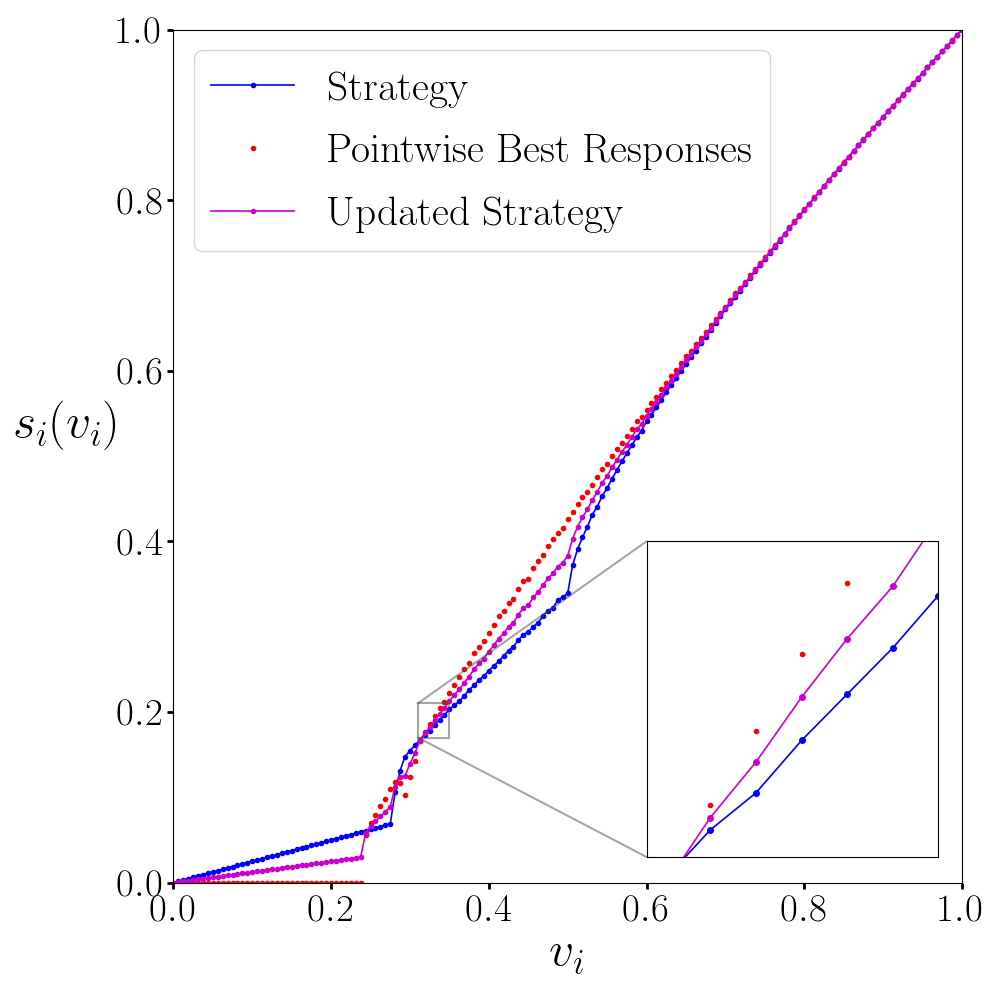}
\caption{Search phase.}
\label{fig:LLG_BR}
\end{subfigure}
\enskip
\begin{subfigure}{0.49\textwidth}
\includegraphics[width=\textwidth]{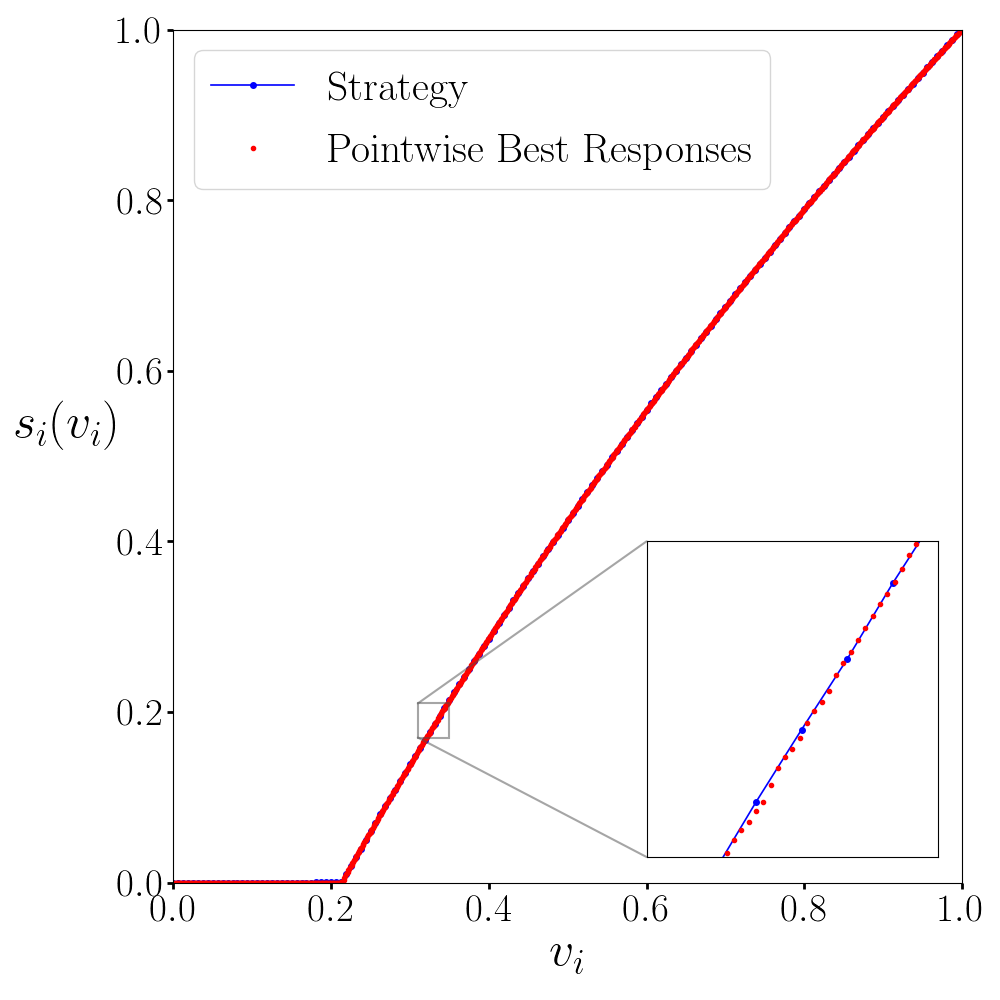}
\caption{Verification phase.}
\label{fig:LLG_Verif}
\end{subfigure}
\caption{Illustration of piecewise linear interpolation as used in our algorithm when applied to \LLG, with the proxy payment rule,  $\alpha=1.0$, and $\gamma=0.5$. (a) Strategies after two iterations (shown in blue) and three iterations (shown in purple), with pointwise best responses  used to perform the update (shown in red). (b) Strategy after eleven iterations, with pointwise best responses used for verification.}
\end{figure}

{\color{black}
To compute best responses efficiently, we construct bidders' strategies using linear interpolation.
The strategy $s_i$ is constructed by specifying a series of \emph{control points}, which are simply valuations (e.g., elements of $[0,1]$ in the case of LLG), and then assigning a bid to each control point.
To extend this to a strategy over the full value space, we perform linear interpolation between the bids at neighboring control points.\footnote{\color{black} For strategies in more than one dimension, we place the control points on a grid and perform multilinear interpolation between them. The grid structure of the control points will be helpful for applying our theoretical results (Section~\ref{sec:provenErrorBound}), which is why we prefer this method over alternatives such as linear interpolation on an arbitrary triangulation.}
There are many other methods that one might use to interpolate strategies (e.g., piecewise constant interpolation, splines, etc); we find piecewise linear interpolation to be particularly attractive as it is simple (and thus fast to evaluate) but the resulting strategy can approximate any bounded function well, given a sufficient number of control points.
}

Using piecewise linear interpolation, the task of constructing the best response $BR_i$ given in Equation~\eqref{eq:brExp}  simplifies to  finding the bid $BR_i(v_i)$ for each control point $v_i$, which we call the \emph{pointwise best response} at $v_i$.
We will discuss the subproblem of finding pointwise best responses in the next subsection.
{\color{black}
Note that this algorithm design implies that there are many valuations for which a bidder's expected utility is not directly optimized.
However, this is not problematic since we later \emph{verify} that the final $\varepsilon$-BNE is valid for \emph{all} possible valuations.
}

An example of a strategy update for the \LLG domain is shown in Figure~\ref{fig:LLG_BR}. There we  show the pointwise best responses, the dampened update of the bids at each control point, and  the linear interpolation to the full strategy.
\subsection{Pointwise Best Responses}
\label{sec:pwbr}

To compute the best response $BR_i$ over the whole value space, we need to compute many \emph{pointwise} best responses $BR_i(v_i)$, i.e., find the bid $b_i$ that maximizes the expected utility $\eu_i(v_i, b_i)$ for a fixed valuation $v_i$.
This means that at each such $v_i$, we must solve an unconstrained optimization problem over the whole action space (i.e., the set of possible bids).
This is computationally expensive and can only be approximated with numerical algorithms that evaluate $\eu_i(v_i, b_{i})$ at many different bids $b_i$.
Furthermore, the function to be optimized may be non-convex and/or non-differentiable in $b_i$ (as discussed in Section \ref{sec:bayesiangames}). There are no known algorithms that can solve non-convex optimization problems to global optimality in finite computation time.
For this reason, we employ  a sophisticated version of \emph{pattern search} (see Section \ref{sec:PatternSearch})  which can quickly find highly accurate best responses.\footnote{\citet{vorobeychik2008stochastic} proposed a best response procedure with global convergence guarantees that can be applied to auctions. However, because their procedure is not guaranteed to converge in finite time, we do not employ it in our algorithm.} 

Evaluating the expected utility $\bar{u}_i(v_i, b_i)$ is already a challenge for a fixed $v_i$ and $b_i$ because it involves solving the following  high-dimensional integral:
\begin{equation}
  \label{eq:brIntegral}
  \bar{u}_i(v_i, b_i)
  = \int_{v_\smi \sim \mathcal{V}_{\smi} | V_i = v_i} u_i(v_i, b_i, s_\smi(v_\smi)) f_\smi(v_\smi) \diff v_\smi.
\end{equation}
Here,  $f_\smi$ is the joint PDF associated with the distribution of $v_\smi$.
The dimensionality of the integral is $(n-1)\cdot d$ for an auction with $n$ bidders and value space dimension $d$. Consequently, to solve this integral, we  use \emph{Monte Carlo (MC) integration}, a technique that works by averaging the integrand $u_i(v_i, b_i, s_\smi(v_\smi)) f_\smi(v_\smi)$ over many samples of the variable of integration $v_\smi$, taken uniformly at random over all its possible realizations.\footnote{Note that Monte Carlo integration only computes an \textit{estimate} of the true value of the integral. We could use a concentration inequality (e.g., \citet{hoeffding1963probability}) to derive a confidence interval around this estimate. In our experiments, we choose a very large number of Monte Carlo samples such that the estimate becomes very accurate (by the central limit theorem); given this, we forego computing confidence intervals, following standard practice using Monte Carlo integration. However, in Section \ref{sec:compAnalyticResults}, we provide an analysis of the numerical robustness of the overall algorithm.}

\begin{remark}
  In the simple \LLG setting, the integral in Equation (\ref{eq:brIntegral})  is only two-dimensional and thus may  alternatively be solved via numerical quadrature in less time.
  In larger domains, however, only Monte Carlo techniques will scale efficiently with the problem dimension.
To keep the presentation of the algorithmic techniques comparable throughout the paper, we exclusively evaluate their performance using Monte Carlo integration.
Our released source code includes both implementations.
\end{remark}

\subsection{Verification}
\label{sec:estimatingError}

In the search phase, we are free to simplify the strategy space or make use of any other heuristics to speed up the search.
As a consequence, we only have an estimate of the $\epsilon$ of our current strategy profile, i.e., the maximum of the utility loss $\loss_i$ at all control points.
This estimate is precise enough to decide when to break out of the search phase. But to know that we have found an $\epsilon$-BNE strategy profile, we need to make sure that the utility loss is indeed less than $\epsilon$ at \emph{all} valuations.

 At first sight, it is not obvious how to achieve this, given that the value space is continuous, and we can thus not simply evaluate the utility loss at every possible valuation.
To address this problem, we have developed Theorem~\ref{thm:epsbound2} which enables us to bound the epsilon over the \emph{full} value space using only a finite number of \emph{verification points} (as illustrated in Figure~\ref{fig:LLG_Verif}). We devote Section \ref{sec:provenErrorBound} to a detailed explanation of this approach. For now, we highlight that the verification phase of the algorithm (Lines 16-17) is significantly enhanced by this theory. Note that in those settings where the conditions of the theorem are not met, we fall back on an alternative verification method that computes an estimate of the $\epsilon$ (Section \ref{sec:estimated_epsilon}).

In the next section, we describe the design of the search phase in detail.
Note that our experimental setup in that section will already require us to be able to compute the $\epsilon$ of a given strategy profile using a verification method.
For the moment, we will simply treat the verification phase as a black box that provides us with the $\epsilon$ for  any given strategy profile.
Then, in Section~\ref{sec:provenErrorBound}, we will explain our verification methods in detail.
\section{The Search Phase}
\label{sec:compBestResponses}

In this section, we turn our attention towards the details of the search phase.
Recall that in the search phase, our goal is to find an $\epsilon$-BNE as \textit{quickly} as possible.
For this, it is important to observe that the runtime of each iteration of our inner loop depends on three key parameters:
(1) the number of Monte Carlo samples used for integrating the expected utility,
(2) the number of times the expected utility is evaluated to find each pointwise best response, and
(3) the number of control points at which the pointwise best response is computed.
To achieve a fast runtime of our algorithm, these parameters need to be kept as small as possible, while ensuring that we are still able to find $\epsilon$-BNEs across a wide variety of settings. This whole section shows how to achieve this. Towards this end, we first present a baseline algorithm for computing best responses, and then provide a series of algorithmic improvements, explaining and testing each of them in turn. The application of all of our techniques leads to a cumulative \inref{SpeedAdPt}-fold speedup of our algorithm over the baseline, allowing us to find $\epsilon$-BNEs in only a few seconds in \LLG.

\subsection{Experimental Set-up}

We evaluate the runtime of each algorithmic improvement on the test set of \numSettings different variations of the \LLG domain (as introduced in Section~\ref{sec:llg}).
We use this test set because finding $\epsilon$-BNEs is still tractable in this small setting, and it enables us to compare our results against known analytical BNEs, as we will do in Section~\ref{sec:compAnalyticResults}.
We run each version of our algorithm \numRuns times on each of the \numSettings auction settings, and we report the average of these \numSettingsTimesRuns runs.\footnote{To simplify replicability, we use a different but fixed random seed for each of the runs.
This set-up makes our experiments deterministic (in the sense that they are perfectly repeatable) while still capturing the effects of randomness on our algorithm's runtime.} Runtime results for finding a $\targetEpsilon$-BNE are presented in Table~\ref{tab:LLGRuntime}.
Each run is performed single-threaded on a 2.8Ghz Intel Xeon E5-2680 v2.

Note that several of our techniques make a trade-off between
speed and accuracy, so it is important to evaluate their effectiveness \emph{as a whole} to capture how changes in accuracy affect the convergence rate and thus the speed of the overall algorithm.
Therefore, we do not just measure a single best response calculation in isolation, but measure the runtime of the entire algorithm from its start at the truthful strategy profile until reaching convergence.
However, the runtime of the whole algorithm is also affected by (a) how the transition between search and verification is done, and (b) the runtime of the verification phase itself. 
This introduces additional ambiguity regarding the true performance of our search techniques. For this reason, in this section, we separate these effects by only measuring the runtime of the inner loop of Algorithm \ref{alg:IBR} (i.e., the core of the search phase).
We do this by skipping the convergence check in Line 9 and running the inner loop for a large, fixed number of iterations instead.
For each of the strategy profiles visited by this process, we then run our verification phase to determine the first iteration at which $\epsilon$ falls below the threshold of $\targetEpsilon$ (our target $\epsilon$), at which point we consider the search to have converged.
Importantly, the runtime numbers we report in this section exclude the time needed for verification.
To understand the performance of the full algorithm, we will perform an end-to-end analysis in Section~\ref{sec:puttingitalltogether}, including the outer loop of the search phase and the verification phase.

\subsection{Naive Monte Carlo Sampling and the Problem of Variance}

We first present a straightforward implementation of our BNE algorithm.
We lay an evenly-spaced grid of \numCtrlPtsNotAdaptive control points over the value space,
and at each control point $v_i$, we maximize $\eu_i(v_i, .)$ by running Brent search, a commonly-used form of unconstrained optimization over the space of possible bids  \cite{brent1971algorithm}.
Each evaluation of $\eu_i(v_i, b_i)$ is done via Monte Carlo integration as described in Section~\ref{sec:pwbr}.

This version of the algorithm is a reasonable first step, but it fails to converge to our target $\varepsilon=\targetEpsilon$, even when using $\maxSamplesNaive$ Monte Carlo samples.
The reason for this is that the computation of the expected utility $\eu_i$ has very high variance.
In MC integration methods, any reduction in variance is always desirable, of course, but in our application this consideration is especially important.
In an iterated best response algorithm, computing an equilibrium is fundamentally a dynamic process, where the output of one iteration is fed as input into the next.
When we have high variance in the expected utility computation, this causes the computed best response to deviate from the true best response in a random direction at each control point.
Since an $\varepsilon$-BNE is defined by the worst-case utility
loss over all valuations of all bidders, a large error at a single control point during the best response computation prevents the entire algorithm from converging.
This can produce the counter-intuitive
effect that increasing the number of control points actually decreases the accuracy of the algorithm.

Our first two algorithmic improvements address this issue, reducing the variance of the Monte Carlo integration to acceptable levels.
The third version of the algorithm is the first one that converges to our target $\epsilon$, so we will use that one as the baseline algorithm against which we measure the performance impact of later improvements.

\definecolor{tablehighlight}{rgb}{0.9,0.9, 0.9}

\begin{table}[tb]
\newcommand{\numtable}[2]{\customlabeloutput{\num[round-precision={#1}]{#2}}}
\newcommand{\numlabel}[2]{\customlabel{#1}{\num[round-precision=1]{#2}}\numtable{1}{#2}}
\centering
\setlength\tabcolsep{2pt}
\begin{tabular}{l||r|r|r}
\multicolumn{1}{c||}{\textbf{Algorithm}} &
\multicolumn{1}{c|} {\crcell{\textbf{Average}\\\textbf{Iterations}}} &
\multicolumn{1}{c|} {\crcell{\textbf{Average Runtime}\\\textbf{in Seconds}}} &
\multicolumn{1}{c } {\crcell{\textbf{Cumulative}\\\textbf{Speedup}}} \\
\hhline{=::=:=:=}
Naive Monte Carlo  & - & -  & - \\
+ Importance Sampling  & - & -  & - \\
\rowcolor{tablehighlight}
+ Quasi-Random Numbers (baseline)  & \numlabel{ItersQuasi}{8.869} (\numtable{3}{0.0818}) & \numlabel{RuntimeQuasi}{950.478} (\numtable{2}{10.94})   & -                              \\
+ Common Random Numbers    & \numtable{1}{8.438} (\numtable{3}{0.0729}) & \numtable{1}{26.66} (\numtable{2}{0.3676})    & \numlabel{SpeedCRand}{35.65}x  \\
+ Adaptive Dampening   & \numtable{1}{5.893} (\numtable{3}{0.0488}) & \numtable{1}{18.744} (\numtable{2}{0.2382})    & \numlabel{SpeedAdDamp}{50.7}x \\
+ Pattern Search       & \numtable{1}{5.875} (\numtable{3}{0.0482}) & \numtable{1}{15.046} (\numtable{2}{0.14})     & \numlabel{SpeedPat}{63.17}x    \\
+ Adaptive Control Points & \numtable{1}{6.062} (\numtable{3}{0.0523})   & \numtable{1}{4.681} (\numtable{2}{0.0367})     & \numlabel{SpeedAdPt}{203.05}x  \\
\end{tabular}
\vspace{-0.05in}
\caption{Runtimes for the inner loop of our algorithm, with standard errors shown in paretheses. We compare the time required for different variants of our algorithm to achieve an estimated $\targetEpsilon$-BNE, averaged over \numRuns runs in each of our \numSettings auction settings. Note that the first two variants of the algorithm did not converge on all \numSettingsTimesRuns instances.}
\label{tab:LLGRuntime}
\vspace{-0.12in}
\end{table}

\newcommand{\SpeedAdDampRel}  {\customlabeloutput{1.42}} 
\newcommand{\SpeedPatRel}     {\customlabeloutput{1.24}} 
\newcommand{\SpeedAdPtRel}    {\customlabeloutput{3.21}} 

\subsection{Importance Sampling}

{\color{black}
When using MC integration to determine the expected utility of bidder $i$, it is often desirable to adapt the sampling process to exclude realizations of other bidders' valuations $V_\smi$ that lead bidder $i$ to win the empty bundle.
This is because, in most auctions, the utility in this situation is simply zero, and the computational effort used to evaluate that sample is wasted without providing any additional information about the expected utility.
Excluding such valuations can be achieved by implementing a variant of \emph{importance sampling} \cite{press2007numerical}, where we truncate the distribution that we sample from, and correct for this change by appropriately weighting the samples obtained.

In general, importance sampling is difficult to implement in CAs, as it requires characterizing the realizations of $V_\smi$ where bidder $i$ wins the empty bundle.
Thus, the implementation of this technique is domain specific, as it must exploit the properties of the domain at hand to be efficient.
In \LLG, we know that local bidder 1 wins exactly when $b_1 + s_2(v_2) > s_3(v_3)$. Therefore, we can provide the following implementation in LLG: we first draw the valuation $v_2$ from $\Unif[0, 1]$ as usual, but then we draw the valuation $v_3$ from $\Unif[0, b_1 + s_2(v_2)]$ instead of from $\Unif[0,2]$.
This guarantees that the local bidders are winners, because we know that the global bidder is truthful and thus $s_3(v_3) = v_3$.
The resulting sample is then multiplied with the factor $(b_1 + s_2(v_2)) / 2$ to correct the bias we have introduced.
This technique is especially important for small bids where bidder 1 only wins rarely, where otherwise very few samples would fall in the winning region, leading to high variance.
We have added this LLG-specific implementation to our algorithm, but we leave the exploration of more general implementations to future work.
}

While highly useful, even with this improvement our algorithm still fails to converge on all instances, so this version of the algorithm is also not suitable as a baseline.
For that, we need to add another improvement.
\subsection{Quasi-Random Numbers}

Another effective method for reducing variance is to replace \emph{pseudo-random} numbers (i.e., the standard random number implementation available in most programming languages) with \emph{quasi-random} numbers in the sampling process.
While pseudo-random numbers try to reproduce the properties of true random numbers as closely as possible (including their tendency to form clumps), quasi-random numbers cover the sampled region more evenly \cite{morokoff1995quasi}.
In our algorithm, we use a multi-dimensional Sobol sequence. This improvement enables convergence to our target of $\varepsilon=\targetEpsilon$, using \numMCSamplesQuasi MC samples.
On average, the search converges in \inref{ItersQuasi} iterations and \inref{RuntimeQuasi} seconds (see Table~\ref{tab:LLGRuntime}).
As this is the first variant of the algorithm that reaches convergence, we use it as the baseline algorithm against which we compare additional techniques.

\begin{remark}
The number of MC samples required for convergence might seem surprisingly high.
Many of our auction instances would also converge using considerably fewer samples, but \numMCSamplesQuasi is the minimum number required to make each of our 16 auction settings converge in all \numRuns out of \numRuns runs.
{\color{black}
In fact, most of the auction instances converge with as little as 20,000 - 50,000 samples, but some specific instances, in particular those using the proxy payment rule, need larger numbers for convergence.
We have also observed a small trend for instances with non-uniform value distributions ($\alpha = 2$) to require more samples for convergence than their uniform counterparts.
}
To keep our experimental set-up simple and consistent, we did not optimize the number of samples for each setting individually, but rather chose a single number that makes our full set of \numSettingsTimesRuns runs converge.
\end{remark}
\subsection{Common Random Numbers}

When computing a  pointwise best response for a given valuation, we repeatedly compare the expected utility of two different bids.  If $X$ and $Y$ are the random variables representing the expected utility associated with two bids, then we want to determine if $\Exp [ X ] - \Exp [ Y ]$ is greater
or smaller than zero to decide which bid is better. Using common random numbers \cite{glasserman1992some}, we can compute
$\Exp [ X - Y ]$ instead and get the same result with
lower variance.
We integrate this idea by using the same sequence of samples to
compute both $\Exp [ X ]$ and $\Exp [ Y ]$.  The samples used for both integrals are pairwise
perfectly correlated but still quasi-random when considering each of
the integrals in isolation.  
{\color{black}
Adding this technique improves the convergence behavior of our algorithm substantially.
Correspondingly, we reduce the number of samples used to \numMCSamplesCommon, which is enough for all runs to converge to our target $\varepsilon$. This  is only \pctMCSamplesCommonToBaseline of the samples needed by the baseline algorithm, resulting in a \inref{SpeedCRand}-fold speedup.
}
Note that we get more than a \factorMCSamplesCommonToBaseline-fold speedup because, in addition to saving a factor \factorMCSamplesCommonToBaseline in the expected utility computation, this change decreases the number of function evaluations required by Brent search for finding the best response, and additionally makes the algorithm converge in slightly fewer iterations.
\subsection{Adaptive Dampening of Strategy Updates}

As mentioned in Section~\ref{sec:bneAlgorithmFramework}, in each iteration of the inner loop we perform a dampened update of the strategy profile.
This means that instead of just replacing the previous strategy with its best response (as would be done in a naive implementation of iterated best response), we make the current strategy at each control point $v_i$ a weighted combination of the previous strategy and the best response: $s_i(v_i) = (1-w) \cdot s_i (v_i) + w \cdot BR_i(v_i, s_{\smi})$, for some update weight $w \in [0,1]$.
However, if $w$ is chosen too high, the iterative process might fall into an oscillating behavior around the solution and converge very slowly or not at all.
Conversely, if $w$ is chosen too low, convergence is slowed down when the current strategy is far from equilibrium.
In the baseline algorithm, we use an update weight $w=\dampeningFactorNaive$ to balance these two concerns, but of course a fixed update weight is suboptimal.

To make the convergence process of the algorithm smoother, we now design an adaptive form of strategy updates (see \citet{Fudenberg1995FictitiousPlay} and \citet{lubin2009quantifying} for earlier work on adaptive dampening).
Concretely, we set the update weight dynamically for each individual control point, based on how close to a solution we expect to be:
\begin{equation}
    \label{eq:dampening}
  w(v_i) = \frac{2}{\pi} \arctan(c \cdot \loss_i(v_i, s_i(v_i))) \cdot (w_{max} - w_{min}) + w_{min}.
\end{equation}

To understand the effect of the rule, observe that $\frac{2}{\pi} \arctan$ maps all positive numbers into the $[0, 1)$ interval, and therefore Equation~\eqref{eq:dampening} creates a weight between $w_{min}$ and $w_{max}$, separately for each
control point.\footnote{In our experiments we use the constants $w_{min} = 0.2$, $w_{max} = 0.7$, and $c = \frac{1}{2\epsilon}$.} Thus, when the utility loss $\loss_i$ at the control point $v_i$ is large, the weight $w(v_i)$ is also large, resulting in an aggressive update step that reaches the neighborhood of an equilibrium in few iterations.
Conversely, when the utility loss is small, the update step is much more conservative, reducing the risk of overshooting the equilibrium.
This adaptive dampening technique combines the benefits of a rule with a fixed large weight (which allows taking big steps) and a rule with a fixed small weight (which leads to high stability). Adding  adaptive dampening results in another \SpeedAdDampRel-fold speedup, leading to a cumulative \inref{SpeedAdDamp}-fold speedup over
the baseline algorithm.
\subsection{Pattern Search}
\label{sec:PatternSearch}

{\color{black}
So far, we have used Brent search to compute pointwise best responses. While Brent search is a very accurate method for function maximization, it has two problems.
First, it evaluates the function to be maximized many times.
Second, it can only be applied to single-dimensional functions and thus prevents us from applying our algorithm to larger, multi-dimensional auction settings (including those we will study in Section~\ref{sec:HigherDims}).
What we need is an optimization procedure that requires fewer function evaluations and scales to an arbitrary number of dimensions.

To this end, we replace Brent search with a slightly modified version of   \emph{pattern search} \cite{hooke1961direct}. 
Pattern search is a type of hierarchical local search that evaluates a number of points around a \textit{center} according to a fixed
pattern. The center is initialized to a guess of the solution (in our case, we use the current strategy at the given control point).  Whenever a better solution is found, pattern search moves the center of the pattern there and continues searching.  If a better solution could not be found, it reduces the size of the pattern by half and continues searching around the center at a smaller scale.
Thus, there are two kinds of steps that can occur during the search: \textit{moving steps} and \textit{shrinking steps}.
Standard pattern search terminates when the pattern reaches a sufficiently small scale.
However, choosing the correct scale to stop the search is particularly challenging in our case, which we resolve as follows.
Recall that, in our BNE algorithm, each bidder's strategy is repeatedly updated to be a best response to the other bidders' previous strategies.
In the early iterations of this process, high precision is not needed, because we are far away from an $\epsilon$-BNE, and thus very small differences in utilities can be ignored because strategies will undergo major changes from one iteration to the next anyway.
In contrast, once the iterative process has roughly settled on the final strategies, high precision is required to actually converge with the $\varepsilon$ we are aiming for.

To properly serve both of these needs, we slightly modify pattern search to choose  its precision adaptively. Concretely, we equip the algorithm with a budget that denotes the maximal number of steps the algorithm is allowed to take. We subtract two steps from the budget whenever the algorithm takes a moving step, and only one step whenever the algorithm takes a shrinking step. This has the effect of only allowing the algorithm to take many steps when pattern search stays close to the initial strategy at the given control point (which is an indicator that we are close to convergence). In \LLG, we use a pattern of 3 points and a budget of 12 steps.} Overall, this \textit{budgeted pattern search} technique  requires fewer evaluations of the expected utility than Brent search when
high precision is unnecessary, and is almost as accurate when high precision
\emph{is} needed.
Adding budgeted pattern search to our algorithm results in a \SpeedPatRel-fold speedup, leading to a cumulative \inref{SpeedPat}-fold speedup over the
baseline algorithm.
\subsection{Adaptive Control Point Placement}
\label{sec:adaptiveCP}

{\color{black}
A bidder's BNE strategy often has regions of both high and low
curvature (see, e.g., Figure~\ref{fig:LLG_Verif} in Section~\ref{sec:bneAlgorithmFramework}). Thus, when constructing a piecewise linear function to approximate this BNE strategy, using an equal spacing of control points either has too many unnecessary points in flat regions (wasting computational effort) or not enough points in curved regions\ (not achieving sufficient accuracy when needed). In particular, if we  used a spacing of control points that is too coarse, then the linear interpolation between the control points may not approximate the true best response well enough in those regions where the best response
is strongly curved (or even non-smooth). This can lead to a large utility loss at \emph{some} valuations in the value space (and thus a large $\varepsilon$), even when the utility loss  at the control points themselves  is very small.

To avoid this, we use an adaptive control point placement method which iteratively constructs an irregular grid with a higher density of control points in regions where this is deemed necessary. This is a well-known technique \cite{huang2010adaptive}; however, applying it to a specific problem requires specifying a criterion by which to determine the highest priority region where the next control point should be placed. In LLG, we initialize this method with an evenly spaced grid of only 10 control points.
We then iteratively place additional control points at the midpoint between two neighboring  control points where the estimated curvature of the best response function is largest. %
Using this adaptive method, the control points are spaced further apart in regions of low curvature, allowing us to reduce the overall amount of work done in the best response computation, while retaining the same accuracy where needed. We provide pseudocode for this method in Appendix~\ref{app:adaptiveCP}.}

Using adaptive control point placement, we obtain convergence in all \numSettings  auction settings with only $\numCtrlPtsAdaptive$ control points (10 in the initial grid plus 30 adaptively placed) instead of the $\numCtrlPtsNotAdaptive$ previously required.
Adding adaptive control points to our algorithm results in a \SpeedAdPtRel-fold speedup, leading to an overall \inref{SpeedAdPt}-fold speedup.
The cumulative improvement achieved by adding all of the techniques in this section means that our algorithm can find very precise $\epsilon$-BNEs in the LLG domain using only a few seconds of computation time.
\section{The Verification Phase}
\label{sec:provenErrorBound}

Now that we have a fast algorithm to find an $\epsilon$-BNE candidate, it remains to show how to accurately compute $\epsilon$.
This is the task of the verification phase.
We have already established that the utility loss is small \textit{at each control point} used in the search phase; otherwise, the algorithm would not have broken out of the iterated best response loop.
However, a computation of $\epsilon$ requires us to find the maximum utility loss \textit{at every point} in the continuum of valuations.
Unfortunately, the maximum is not a well-behaved statistic, in the sense that even if we were to check the utility loss at more and more control points, we would never be sure that there is no gap remaining between the worst utility loss we have observed and the worst utility loss that exists.\footnote{This is a different dynamic than, e.g., computing the \emph{average} of a function, where taking more and more samples guarantees convergence towards the exact solution, even in continuous domains.}
This is why we need to think carefully about the design of the verification phase.

On a high level, our verification method (and main technical contribution of the paper) involves reasoning about the utility loss in terms of intervals $[v_i, v_i']$ of valuations. The key insight lies in using certain properties of the endpoints $v_i$ and $v_i'$ to provide an upper bound that holds for any valuation contained in the interval.
This approach is almost fully general: it works for auctions with any allocation and payment rule, but we must make a few assumptions on the bidders' valuations and utilities.
We formalize our results as Theorems~\ref{thm:epsbound}~and~\ref{thm:epsbound2}  in Section~\ref{sec:theorem},  and we then discuss how to compute the upper bound on $\epsilon$ based on these theorems in Section~\ref{sec:theoremApplication}.\footnote{Note that Theorems~\ref{thm:epsbound} and \ref{thm:epsbound2} are significantly stronger versions of the theorem we presented in \cite{bosshard2017fastBNE} producing much tighter bounds on $\epsilon$.  The general idea of using a finite subset of the value space to obtain a bound on the whole value space was also used more recently by \citet{Balcan2019} who employ learning theory to estimate approximate incentive compatibility of non-truthful mechanisms. }

For those settings where the assumptions of Theorems \ref{thm:epsbound} and \ref{thm:epsbound2} are not satisfied, we fall back on an alternative verification method that computes an estimate of $\epsilon$ instead of an upper bound. We present this alternative verification method in Section~\ref{sec:estimated_epsilon}. In  Section~\ref{sec:upper_vs_lower_bound}, we show that, when using sufficient computation time, the upper bound on $\epsilon$ as well as the estimated $\epsilon$ are very close to the true $\epsilon$.

\subsection{Deriving an Upper Bound on $\bm\varepsilon$}
\label{sec:theorem}

Our upper bound on $\epsilon$ requires three properties to hold regarding the bidders in the CA:

\begin{assumption}[Linear Utilities]
Each bidder's utility function is of the form
\begin{equation}
    u_i(v_i, b, X, p) = v_i(X_i(b)) - p_i(b),
\end{equation}
i.e., linear in the valuation $v_i$.
\label{asm:QuasiLinear}
\end{assumption}

\begin{assumption}[Bounded Value Spaces]
{\color{black}
Each bidder's value space $\mathbb{V}_i$ is bounded.
}
\label{asm:BoundedValues}
\end{assumption}

\begin{assumption}[Independently Distributed Valuations]
The valuations $V_i$ are mutually independent random variables.
\label{asm:IndependentValues}
\end{assumption}

Assumption~\ref{asm:QuasiLinear} is relatively standard in auction theory.
This class of utility functions is a subset of quasi-linear utilities \cite{Mas-Colell1995Mic}; it excludes some bidder models that are less common, e.g., quasi-linear utilities with risk aversion.
Assumption~\ref{asm:BoundedValues} is also standard and not very restrictive.
In contrast, Assumption~\ref{asm:IndependentValues}
is more restrictive.
It excludes all CAs with interdependent valuations (which includes settings where bidders' private values depend on publicly observable signals).

To derive a bound on $\epsilon$, we need to know the gap between the expected utility $\eu_i(v_i, b_i)$ obtained at an equilibrium candidate $s$, and the expected utility $\bru_i(v_i)$ that can be obtained by playing a best response against $s_\smi$.
To bound the gap between these two functions, we now make two observations.
First, the fact that utilities are linear in the valuation implies that the \emph{expected} utilities are linear in the valuation as well:
\begin{lemma}
    In a CA satisfying Assumption~\ref{asm:QuasiLinear} (linear utilities) and Assumption~\ref{asm:IndependentValues} (independently distributed valuations), for a fixed bid $b_i$, the expected utility $\eu_i(v_{i }, b_i)$ is a linear function in $v_i$.
    \label{lem:utilplane}
\end{lemma}

\begin{proof}
    This follows directly from linearity of expectation. In detail:
    \begin{align*}
    \eu_i(v_i, b_i)
    =& \Exp_{v_{\smi}\sim V_{\smi} | V_i = v_i} \left[ v_i(X_i(b_i, s_\smi(v_\smi))) - p_i(b_i, s_\smi(v_\smi)) \right]\\
    =& \Exp_{v_{\smi}\sim V_{\smi} | V_i = v_i} \left[ \sum_{K \subseteq M} v_i(K) \cdot [ X_i(b_i, s_\smi(v_\smi)) = K ] - p_i(b_i, s_\smi(v_\smi)) \right]\\
    =& \sum_{K \subseteq M} v_i(K) \cdot \Exp_{v_{\smi}\sim V_{\smi} | V_i = v_i} \left[ X_i(b_i, s_\smi(v_\smi)) = K \right] - \Exp_{v_{\smi}\sim V_{\smi} | V_i = v_i} \left[ p_i(b_i, s_\smi(v_\smi)) \right]\\
    =& \sum_{K \subseteq M} v_i(K) \cdot c_i^K - c_i.
\end{align*}
Note that $c_i^K$ and $c_i$ are constants independent of $v_i$ because valuations are independently distributed.
\end{proof}

Second, we show that the best response utility $\bru_i(v_i)$ is convex:

\begin{lemma}
    In a CA satisfying Assumption~\ref{asm:QuasiLinear} (linear utilities) and Assumption~\ref{asm:IndependentValues} (independently distributed valuations), the best response utility $\bru_i(v_i)$ is convex.
    \label{lem:bru_convexity}
\end{lemma}

\begin{proof}
    Per Lemma~\ref{lem:utilplane}, the expected utility $\eu_i(v_{i }, b_i)$ is linear in $v_i$ for a particular bid $b_i$ and therefore also convex.
    The best response utility $\bru_i(v_i)$ is the pointwise supremum of the expected utilities associated with each possible bid $b_i$.
    Therefore, it is the upper envelope of a set of convex functions, which is known to be convex.
\end{proof}

In addition to Assumptions \ref{asm:QuasiLinear}, \ref{asm:BoundedValues} and \ref{asm:IndependentValues}, Theorems~\ref{thm:epsbound} and~\ref{thm:epsbound2}  require the use of \emph{piecewise constant strategies}.
{\color{black}
In Section~\ref{sec:theoremApplication}, we provide a detailed description of how our BNE algorithm converts an arbitrary (candidate) strategy profile into a piecewise constant strategy profile at the beginning of the verification phase. 

}
\subsubsection{Upper Bound in One Dimension}

As a first step, we state a theorem for the upper bound for one-dimensional value spaces, i.e., when bidders are single-minded.
This will help in developing the necessary intuition.
Afterwards, we will state our result for the general case and provide a formal proof.

In the one-dimensional case, Assumption~\ref{asm:BoundedValues}  (bounded value spaces) implies that there are two valuations $v_i^\text{min}$ and $v_i^\text{max}$, such that $\Pr \left[ V_i < v_i^\text{min} \right] = 0$ and $\Pr \left[ V_i > v_i^\text{max} \right] = 0$.
Then, a \emph{piecewise constant} strategy $s_i$ is uniquely defined by a finite set of grid points $w_i^{1} < w_i^2 < \ldots < w_i^{J}$, and a bid $s_i(w_i^j)$ for each grid point $w_i^j$.
The grid must cover the entire value space, that is $w_i^1 \leq v_i^\text{min}$ and $w_i^J \geq v_i^\text{max}$.
For any valuation $v_i$, we then have that $s_i(v_i) = s_i(w_i^j)$, where $w_i^j$ is the grid point closest to $v_i$ from below.

Having defined piecewise constant strategies for the single-minded case, we are now ready to state our first theorem, upper bounding the $\epsilon$ for this case:

\begin{theorem}
\label{thm:epsbound}
{\color{black}
    Consider a CA satisfying Assumption~\ref{asm:QuasiLinear} (linear utilities), Assumption~\ref{asm:BoundedValues} (bounded value spaces), and Assumption~\ref{asm:IndependentValues} (independently distributed valuations).
    Let $s^*$ be a strategy profile where each strategy $s_i^*: \mathbb{R}_{\geq0}\mapsto \mathbb{R}_{\geq0}$ is piecewise constant with grid points $w_i^1, \ldots, w_i^J$.
    Then $s^*$ is an $\varepsilon$-BNE with
    }
    \begin{equation}
        \label{thmformula}
        \epsilon \leq \max_{i \in N} \max_{j \in \{2, \ldots, J\}} \, \, \,
        \max_{w_i \in \{w_i^{j-1}, w_i^j\}} \, \, \,
        \bru_i\left(w_i\right) \, \, - \, \, \eu_i \! \left(w_i, s_i^* \! \left(w_i^{j-1}\right)\right).
    \end{equation}
\end{theorem}
\begin{figure}
\centering

\begin{subfigure}{0.45\textwidth}
\includegraphics[width=\textwidth]{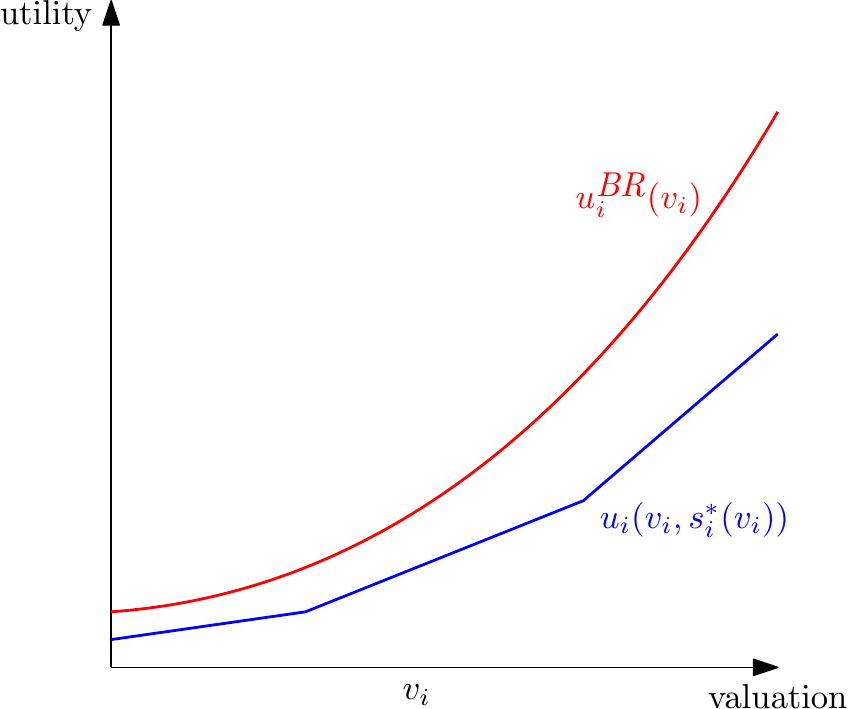}
\caption{Equilibrium utility (blue) and best response utility (red).}
\end{subfigure}
\qquad
\begin{subfigure}{0.45\textwidth}
\includegraphics[width=\textwidth]{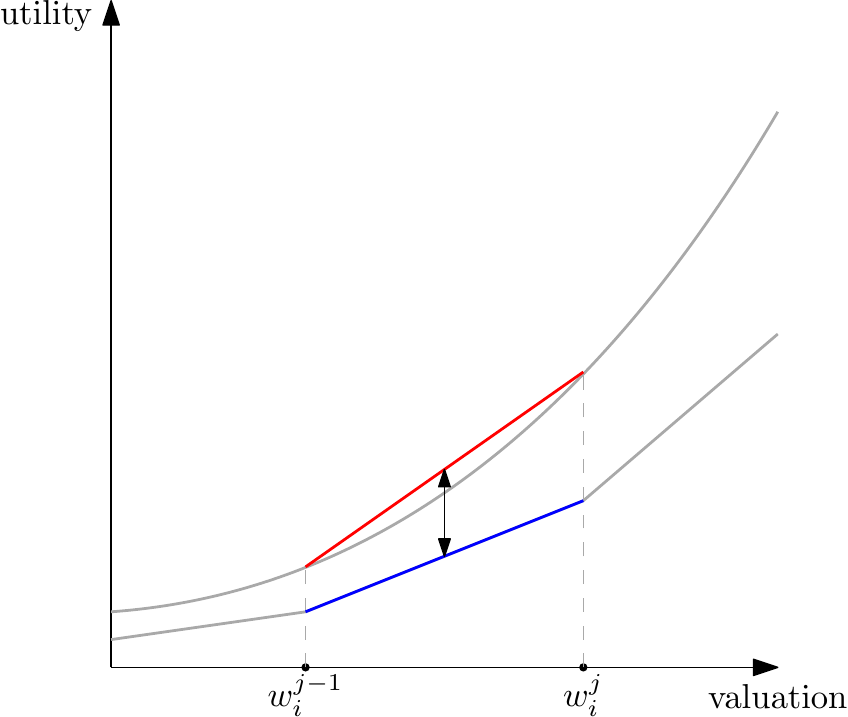}
\caption{The maximum utility loss is bounded by two linear functions.}
\end{subfigure}
\caption{Illustration of Theorem~\ref{thm:epsbound}. Plots are from valuations to utility. Left: the utility that bidder $i$ obtains in equilibrium (which is piecewise linear), plotted against the highest utility he could achieve for any bid (which is convex). Right: The utility loss at any valuation between two control points is at most the difference between two linear functions. The maximum of this difference must be achieved at one of the control points.}

\label{fig:proofbypicture}
\end{figure}

The way the theorem works is as follows: In Equation~\eqref{thmformula}, we bound $\epsilon$ separately for each bidder and each interval between two grid points, with bidders being indexed by $i$ and intervals being indexed by $j$.
For a pair of adjacent grid points $w_i^{j-1}$ and $w_i^j$, we observe two things. First, all values $v_i$ in the interval $[w_i^{j-1}, w_i^j)$ have the same equilibrium bid, namely $s_i^*(w_i^{j-1})$, yielding linearly increasing expected utility by Lemma~\ref{lem:utilplane}.
Thus, the equilibrium utility is a piecewise linear function.
Second, the best response utility is convex in $v_i$ by Lemma~\ref{lem:bru_convexity}, and can thus be bounded from above by the secant between $w_i^{j-1}$ and $w_i^j$.
Figure~\ref{fig:proofbypicture} illustrates the situation graphically.
At any $v_i$, the vertical distance between these two lines is an upper bound for the utility loss. Now, the difference between two linear functions is also linear, and thus achieves its maximum at the boundary of the interval. Therefore, it is enough to check this distance at $w_i^{j-1}$ and $w_i^j$ to compute a bound on the utility loss for all $v_i$ between them.
\subsubsection{Upper Bound in Higher Dimensions}

We now proceed to the general case, where bidders can be arbitrarily multi-minded.
This is more technically involved, as it requires considering the topology of high-dimensional partitions of the value space, and the best response utility must be bounded by a set of simplices.
We begin by introducing the notation required to define piecewise constant strategies in the general case.

\begin{definition}[Cell]
    A cell $ [x, y) \subseteq \mathbb{R}_{\geq0}^d$ is a half-open multi-dimensional interval with lower corner $x$ and upper corner $y$, i.e.,
    $$[x, y) := \bigtimes_{j=1}^d [x_j, y_j).$$
\end{definition}
For convenience, we consider $[x, x)$ to be the singleton $\{x\}$, instead of $\emptyset$.

\begin{definition}[Vertex]
    The vertices of a cell $[x, y)$ are
    \begin{equation}
        \text{Vert}(x, y) := \{x_i + (y_i - x_i)\cdot z \enspace|\enspace z \in \{0,1\}^d\},
    \end{equation}
    where ``$\cdot$'' denotes coordinate-wise vector multiplication.

\end{definition}

\begin{definition}[Cell Partition]
    A cell partition $P_i$ of a bounded value space $\mathbb{V}_i \subseteq \mathbb{R}^d_{\geq0}$ is a set of cells such that each point in $\mathbb{V}_i$  falls into exactly one cell, i.e.,
    $$\forall v_i \in \mathbb{V}_i \, \exists! \, [w_i, w_i') \in P_i \, : \, v_i \in [w_i, w_i').$$
\end{definition}

See Figure~\ref{fig:partitions} for some examples of cell partitions on the 2-dimensional unit cube.
Note that the upper boundary of the cube needs to be covered by cells as well, which is why the partitions include not only 2-dimensional rectangles, but also 1-dimensional lines and even a 0-dimensional point.
{\color{black}
Furthermore, note that cell partitions do not always form regular grids, but can contain cells of different sizes that are not necessarily aligned with each other.
This would be especially relevant if one wanted to adaptively place  control points in higher dimensions, e.g., by iteratively refining a coarse initial partition through subdivision of some of its cells, producing a final partition that is not a grid.
}

\begin{figure}
\centering
\includegraphics[width=0.8\columnwidth]{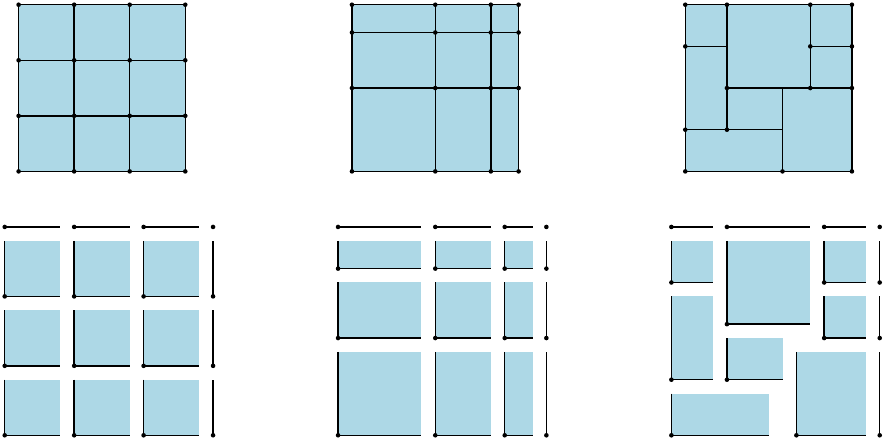}
\caption{Top row: three examples of partitioning the 2-dimensional unit cube. Bottom row: Individual cells of each partition, including lower-dimensional boundary cells. Left: regular grid. Center: non-uniform grid. Right: partition not based on a grid.}
\label{fig:partitions}
\end{figure}

{\color{black}
\begin{definition}[Piecewise Constant Strategy on a Partition]
    Let $P_i$ be a cell partition of bidder $i$'s value space  $\mathbb{V}_i \subseteq \mathbb{R}^d_{\geq0}$.
    A strategy $s_i$ is piecewise constant on $P_i$ if every value has the same bid as the lower corner of the cell it belongs to, i.e.,
    $$\forall [w, w') \in P_i \, \forall v_i \in [w, w') \, : \, s_i(v_i) = s_i(w).$$
    \label{def:PWC}
\end{definition}
}

With all definitions properly in place, we are now ready to state and prove our main result, 
{\color{black}which is an exact multi-dimensional analogue of Theorem~\ref{thm:epsbound}.}

\begin{theorem}
\label{thm:epsbound2}
    {\color{black}
    Consider a CA satisfying Assumption~\ref{asm:QuasiLinear} (linear utilities), Assumption~\ref{asm:BoundedValues} (bounded value spaces), and Assumption~\ref{asm:IndependentValues} (independently distributed valuations).
    Let $s^*$ be a strategy profile where each strategy $s_i^*: \mathbb{R}^d_{\geq0}\mapsto \mathbb{R}_{\geq0}^r$ is piecewise constant on a partition $P_i$ of bidder $i$'s value space $\mathbb{V}_i \subseteq \mathbb{R}^d_{\geq0}$. Then $s^*$ is an $\epsilon$-BNE with
    }
    \begin{equation}
        \label{thmformula2}
        \epsilon \leq \max_{i \in N} \max_{[w_i^\text{min}, w_i^\text{max}) \in P_i} \, \, \,
        \max_{w_i \in \text{Vert}(w_i^\text{min}, w_i^\text{max})} \, \, \,
        \bru_i\left(w_i\right) \, \, - \, \, \eu_i \! \left(w_i, s_i^* \! \left(w_i^\text{min}\right)\right).
    \end{equation}
\end{theorem}

\begin{proof}
    To establish that $s^*$ is an $\varepsilon$-BNE, we need to show that
    \begin{equation}
        \forall i \, \forall v_i : \enskip \epsilon \geq \loss_i(v_i, s_i^*(v_i)).
    \end{equation}
    Consider an arbitrary bidder $i$ and valuation $v_i$. Let $[w_i^\text{min}, w_i^\text{max})$ be the unique cell in $P_i$ with $v_i \in [w_i^\text{min}, w_i^\text{max})$.
    We triangulate $\text{Vert}(w_i^\text{min}, w_i^\text{max})$ with a set of simplices that cover the entire cell \cite{haiman1991simple}.
    This means that $v_i$ is contained in some simplex $Z$ with vertices $\{z_i^1, \ldots, z_i^{r+1}\} \subseteq \text{Vert}(w_i^\text{min}, w_i^\text{max})$.
    We have that $v_i$ is a convex combination of the $z_i^k$, i.e.,
    \begin{equation}
        \exists \lambda^1, \ldots, \lambda^{r+1} \in [0,1]: \sum_{k=1}^{r+1} \lambda^k = 1 \enskip\land\enskip v_i = \sum_{k=1}^{r+1} \lambda^k z_i^k.
    \end{equation}
    It follows from convexity of $\bru_i$ that
    \begin{equation}
        \bru_i(v_i) = \bru_i \left( \sum_{k=1}^{r+1} \lambda^k z_i^k \right) \leq \sum_{k=1}^{r+1} \lambda^k \bru_i(z_i^k) =: f(v_i),
    \end{equation}
    and because strategies are piecewise constant, we have that
    \begin{equation}
        \eu_i(v_i, s_i^*(v_i)) = \eu_i(v_i, s_i^*(w_i^\text{min})) =: g(v_i).
    \end{equation}
    Note that both $f$ and $g$ are linear functions over $Z$.
    The utility loss $\loss_i(v_i, s_i^*(v_i))$ is at most the difference $f(v_i) - g(v_i)$, which is also linear and thus attains its maximum at one of $Z$'s vertices.
    Note that at a vertex $z_i^k$, $f(z_i^k) = \bru_i(z_i^k)$, and by construction $z_i^k \in \text{Vert}(w_i^\text{min}, w_i^\text{max})$.
    Therefore, the term $f(z_i^k) - g(z_i^k)$ corresponding to each vertex of $Z$ is included in the maximization of (\ref{thmformula2}) which implies that the utility loss at each valuation contained in the simplex $Z$ is smaller than $\epsilon$.\end{proof}


{\color{black}

\subsection{Efficiently Computing the Verified $\bm{\epsilon}$}
\label{sec:theoremApplication}

Having proven a theoretical upper bound on $\epsilon$, we now show how our algorithm computes this upper bound in an efficient way in the verification phase. As a first step, this requires us to have a strategy profile with piecewise constant strategies. Our BNE candidate $s$ does not fulfill this requirement, as the strategies are constructed through piecewise \emph{linear} interpolation in the search phase of our algorithm.
To deal with this, Algorithm~\ref{alg:IBR} includes the \convertStrategies subroutine (Line 16), which takes an arbitrary strategy profile and converts all strategies therein into piecewise constant strategies. To convert each strategy $s_i$, we first define a cell partition $P_i$ and evaluate $s_i$ at all cell vertices of this partition. It is important that this partition $P_i$ is based on a   grid, as will become clear in the following paragraph; in our implementation we use a regular grid. We then construct the piecewise constant strategy $s_i^*$ such that the bid for a given valuation $v_i \in [w_i, w_i')$ is equal to $s_i(w_i)$ (i.e., the original bid at the lower corner of the cell that contains $v_i$).
Finally, even though our algorithm searches in the space of piecewise linear strategies, it returns a pair $(s^*, \epsilon)$, where $s^*$ is an $\epsilon$-BNE which only contains piecewise constant strategies. This way, we eliminate any ambiguity about the exact strategy profile to which the $\epsilon$ bound applies.}

{\color{black}
Now that we have a strategy profile $s^*$ consisting of piecewise constant strategies, we next compute the corresponding $\epsilon$ bound according to Equation~\eqref{thmformula2} in Theorem~\ref{thm:epsbound2}. This corresponds to the \verification subroutine (Line 17) in Algorithm~\ref{alg:IBR}. Note that Equation~\eqref{thmformula2} is a maximization over terms of the form:   $\bru_i\left(w_i\right) -  \eu_i \! \left(w_i, s_i^* \! \left(w_i^\text{min}\right)\right)$.  When the value space is $d$-dimensional,  there are $2^d$ such terms per cell. Fortunately, when computing all these terms, we do not need to evaluate each of them individually (which would be very expensive); instead our algorithm only requires \textit{one} pointwise best response computation per grid point. Thus, the number of pointwise best response calls made by our verification method is linear in the number of grid points. Note that this is also the same number of pointwise best response calls made during one full best response iteration (at the same grid size). Concretely, we save on evaluations of terms in Equation~\eqref{thmformula2}  by observing two facts: the first sub-term (i.e., $\bru_i(w_i)$) is repeated many times (once for each cell that has $w_i$ as a vertex), and can thus simply be cached and re-used.
Furthermore, for the second  sub-term (i.e., $\eu_i \! \left(w_i, s_i^* \! \left(w_i^\text{min}\right)\right)$), we can
exploit that  the expected utility is linear in the valuation (Lemma~\ref{lem:utilplane}), to jointly compute this sub-term for all vertices of a cell $[w_i^\text{min}, w_i^\text{max})$.
Specifically, we only compute the expected utility  $\eu_i \! \left(w_i^\text{min}, s_i^* \! \left(w_i^\text{min}\right)\right)$ at a single vertex $w_i^\text{min}$, while now also deriving and storing its linear coefficients (which are just the probabilities of winning each bundle).
The result of this computation can then be easily extrapolated to all other vertices of the cell. Note that our algorithm  requires some small amount of bookkeeping, but  the computational cost of this is insignificant compared to the cost of the pointwise best response computations themselves.

%
}

Note that Theorem~\ref{thm:epsbound2} can be applied no matter how we come up with the $\epsilon$-BNE candidate $s$.
In this paper, we use an iterated best response algorithm, but this could be exchanged for any other equilibrium finding procedure.
For example, our verification step can also be applied to compute the incentives to deviate from truth-telling under a given mechanism, by simply using the truthful strategy profile as the ``$\epsilon$-BNE candidate.''

While we have focused on the application of our results to infinite games, Theorem~\ref{thm:epsbound2} can easily handle finite, but very large value and action spaces. In either case, it is infeasible to perform computations for each individual value or action, which is where our approach shines.
The key ingredient that makes our theory work is that for a fixed action, the payoffs in the auction game are linear in the valuation.

\subsection{An  Estimate of $\bm\varepsilon$}
\label{sec:estimated_epsilon}

Theorems~\ref{thm:epsbound} and~\ref{thm:epsbound2} only hold under certain conditions, given as Assumptions \ref{asm:QuasiLinear}, \ref{asm:BoundedValues} and \ref{asm:IndependentValues} above.
These assumptions do not hold for all CAs that might be of interest: in particular, half of our set of \numSettings \ \LLG settings have correlated valuations (i.e, $\gamma=0.5$), which violates Assumption~\ref{asm:IndependentValues}.
For such auction instances, we fall back on a simpler verification procedure, which consists of computing pointwise best responses on a densely spaced grid of valuations, and estimating $\epsilon$ as the highest utility loss observed among these valuations.
This $\epsilon$ estimate is obviously a lower bound for $\epsilon$.
However, as mentioned earlier, the maximum is not a well-behaved (smooth) statistic, so it is not clear if our $\epsilon$ estimate is close to the true $\epsilon$ or not.

In practice, we have found that this heuristic is a sensible fallback for domains where our theoretical assumptions do not hold.
{\color{black}
The worst case scenario for the heuristic is when two valuations that are very close to each other in value space have very different utility losses.
In the auction games we have studied so far, this situation does not seem to arise.\footnote{\color{black}While we have never observed such a scenario, it is theoretically possible. For some joint value distributions, a small change in bidder $i$'s valuation might imply a very large shift in the marginal distribution of other bidders' bids, due to correlations between bidders' valuations. In this case, $i$'s utility could drastically change, even when $i$ himself does not change his bid. Furthermore, some payment rules might amplify this effect (e.g., by having discontinuities), such that even a small shift in other bidders' marginal distributions can have a large impact on $i$'s utility. Therefore, any bound akin to Theorem~\ref{thm:epsbound2} for correlated settings would need to be domain specific, placing restrictions both on the payment rule and the correlation structure between bidders' valuations.}
}%
We reinforce this claim by comparing both our verification methods with each other in Section~\ref{sec:upper_vs_lower_bound}.

\begin{remark}
Recall that, in the runtime experiments performed in Section~\ref{sec:compBestResponses} in the \LLG domain, we required an implementation of the verification phase as part of our experimental setup (to accurately measure the performance of the search phase).
In those experiments, half the auction settings have correlated valuations, and thus Theorem~\ref{thm:epsbound2} is not applicable.
To keep our experimental set-up simple and consistent, we chose to always use the heuristic method described above to compute an estimated $\epsilon$, even for settings where our theoretical assumptions hold.
For this, we used a grid of $\numCtrlPtsNaiveStopping$ evenly spaced verification points, which is a very large increase in precision when compared to only $\numCtrlPtsAdaptive$ adaptively placed control points used during search.
\end{remark}

\subsection{The Theoretical Bound on $\bm\varepsilon$ vs.\ the Estimate of $\bm\varepsilon$}
\label{sec:upper_vs_lower_bound}

We have now established the theory necessary to get an \emph{upper bound} on the utility loss, but this would not be very useful if the bound was far from the true $\epsilon$.
It is thus natural to ask: how tight is our bound really?
We can answer this question using our $\epsilon$ estimate, which is a \emph{lower bound} on the utility loss, since it is derived by considering a finite subset of all valuations.
We consider the 8 of our 16 \LLG auction settings to which Theorem~\ref{thm:epsbound2} applies, namely those with independently distributed valuations ($\gamma=0$).
For each setting, we take the strategy profile resulting from a full run of our algorithm and compare the lower and upper bounds on $\epsilon$  using the same number of verification points in both cases.

\begin{figure}
\centering
\includegraphics[width=0.7\columnwidth]{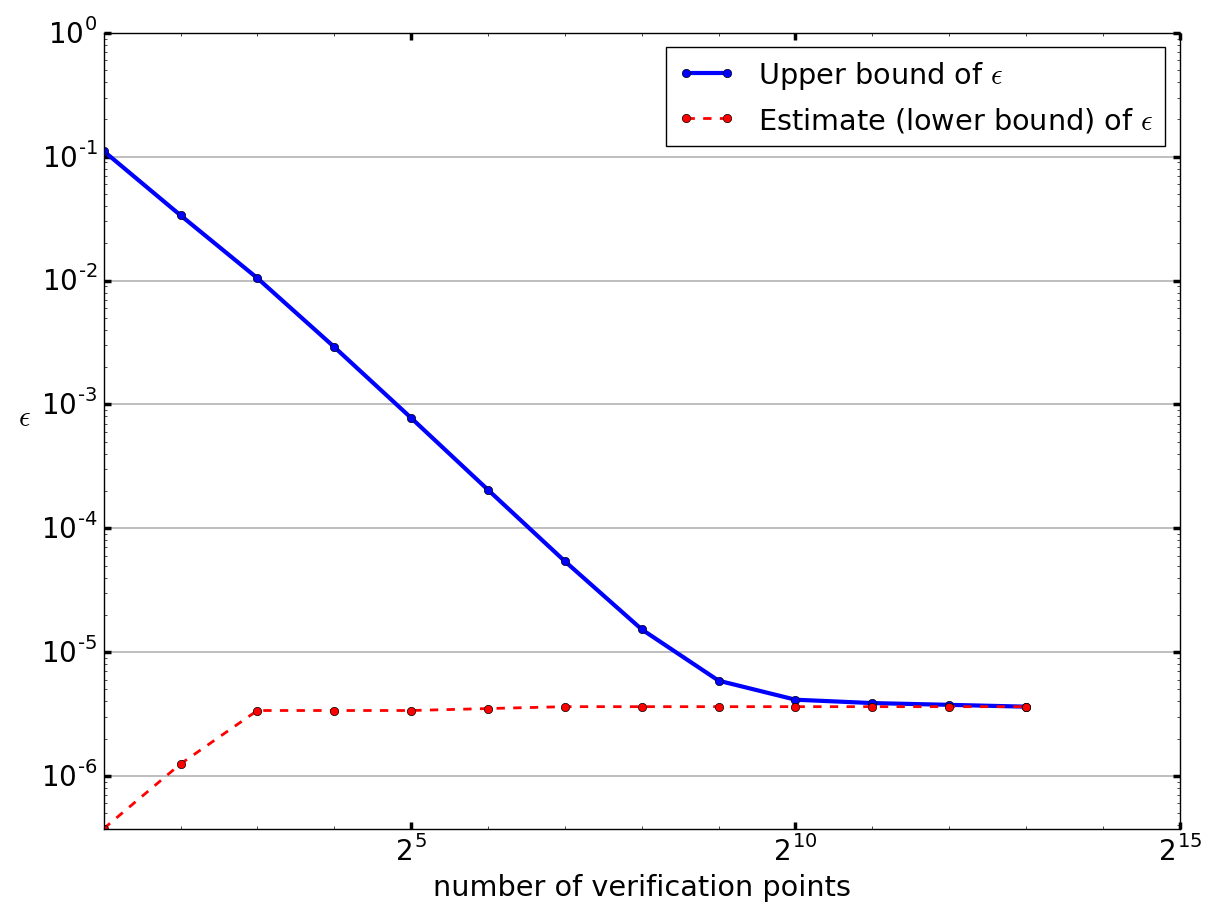}
\caption{Estimate (lower bound) and upper bound for $\epsilon$ of the $\epsilon$-BNEs obtained by our algorithm, averaged over eight different auction settings with uncorrelated bidder valuations. The estimate is computed as described in Section~\ref{sec:estimated_epsilon}, while the upper bound is given by Theorem~\ref{thm:epsbound2}.}
\label{fig:errorplot}
\end{figure}

The result is shown in Figure~\ref{fig:errorplot}, with the number of verification points varying from $2$ to $2^{13}$, and the bounds being the average over the 8 auction settings.
We observe that the lower bound remains practically constant, while the upper bound decreases polynomially in the number of verification points until converging.\footnote{Note that a polynomial relationship corresponds to a straight line on the plot, since it is drawn on a log-log scale.}
At $2^{13}$ verification points, the upper and lower bounds match within machine precision.
Even with as few as $2^9$ verification points, the gap between the two bounds already becomes extremely small.
The results are qualitatively the same when considering each auction separately.

Note that the true $\epsilon$ always lies between the upper and  lower bounds. Because the two bounds converge towards each other, this shows that both bounds are very close to the true $\epsilon$ with sufficiently many verification points. Since we used $\numCtrlPtsNaiveStopping$ verification points in our experiments in Section~\ref{sec:compBestResponses}, this strongly suggests that the $\epsilon$ estimate we reported there was indeed a very good estimate.


\begin{remark}
\color{black}

Note that the accuracy of the $\varepsilon$ obtained by either of our verification methods also depends on how we parameterize the pointwise best response algorithm during the verification phase. 
Concretely, the parameters we have to choose that affect the accuracy of the pointwise best responses are the number of MC samples and the pattern search configuration. Thus, the $\epsilon$ obtained must always be interpreted in light of the parameters used, and any reporting of results must include these parameters.
In our work, we take special care to set the  relevant parameters conservatively.
Specifically, in LLG, we use twice as many Monte Carlo samples for verification as for search (i.e., $20,000$ instead of $10,000$). Furthermore,  we increase the budget of steps for pattern search from $12$ to $20$, which implies that our pattern search can perform up to 20 halvings of the initial pattern size and can thus achieve a search resolution in the action space on the order of $10^{-6}$. In Section~\ref{sec:compAnalyticResults}, we present an empirical convergence analysis, which shows that we set our parameters high enough in relation to our target  $\epsilon$.

\end{remark}
\section{Putting it all Together}
\label{sec:puttingitalltogether}

Now that we have detailed specifications of both the search and verification phases, we can put them together into a full algorithm.
Recall that, for the design of the search phase, we focused on the performance of the inner loop of the algorithm, to keep our experimental set-up simple (Section~\ref{sec:compBestResponses}). We then designed the verification phase (Section \ref{sec:provenErrorBound}). Now we analyze the performance of the algorithm as a whole, which includes the decision regarding when to transition from search to verification.

\subsection{Transitioning from Search to Verification}

        Recall that our BNE algorithm has a target $\epsilon$ it is trying to achieve.
During the search phase, the algorithm only maintains an estimate $\tilde{\epsilon}$ of the true current $\epsilon$, but does not yet have an upper bound on it.
Thus, if the algorithm spends too little computational effort in search and stops early, there is a high risk that the verification phase will show that the target $\epsilon$ has not in fact been reached. On the other hand, if the algorithm expends lots of effort in the search phase, this risk can be minimized. The question thus is how long to search, and when exactly to transition from search to verification.

The amount of computational effort spent is primarily driven by the number of control points used, and this in turn determines to a large degree the accuracy of $\tilde{\epsilon}$  which the search phase uses for its decision regarding when to transition to verification. This motivates our design of two nested loops, with an inner and an outer loop. 
 When the inner loop converges, a best response calculation with higher precision is performed in the outer loop, increasing the accuracy of the strategy update and of the $\tilde{\epsilon}$ calculation.
If this increase in accuracy still leads to an acceptably small $\tilde{\epsilon}$, the algorithm proceeds to verification. Otherwise, the inner loop is resumed, and we require at least two further iterations
before breaking out of the inner loop again.
In this way, the outer loop acts as a gate between the inner loop and the verification phase, only letting through a strategy profile that ``generalizes'' to higher precision parameters. 
For this to work properly, we eschew adaptive control points when in the outer loop, instead using evenly spread out control points densely covering the entire value space.
Furthermore, to account for the lower accuracy of the inner loop, we consider the inner loop to be converged only when an $\tilde{\epsilon}$ less than or equal to $0.8$ times our target $\epsilon$ is reached.

With this transition from search to verification in place, we can now measure the full runtime of the algorithm from start to end, i.e., from the (truthful) starting strategy profile to a verified $\epsilon$-BNE.
For this experiment only, we choose the number of control points in the outer loop to match those in the verification phase.
This simplifies interpreting our runtime results, because this set of parameters achieves  an $\epsilon$ below our target in all \numSettingsTimesRuns out of \numSettingsTimesRuns runs, and thus we don't have to account for the failure rate of the algorithm.
In practice, tuning the configuration of the outer loop involves a classical speed / accuracy trade-off, and it is up to the user of our BNE algorithm to make this trade-off appropriately.

\begin{figure}
\centering
\includegraphics[width=1.0\columnwidth]{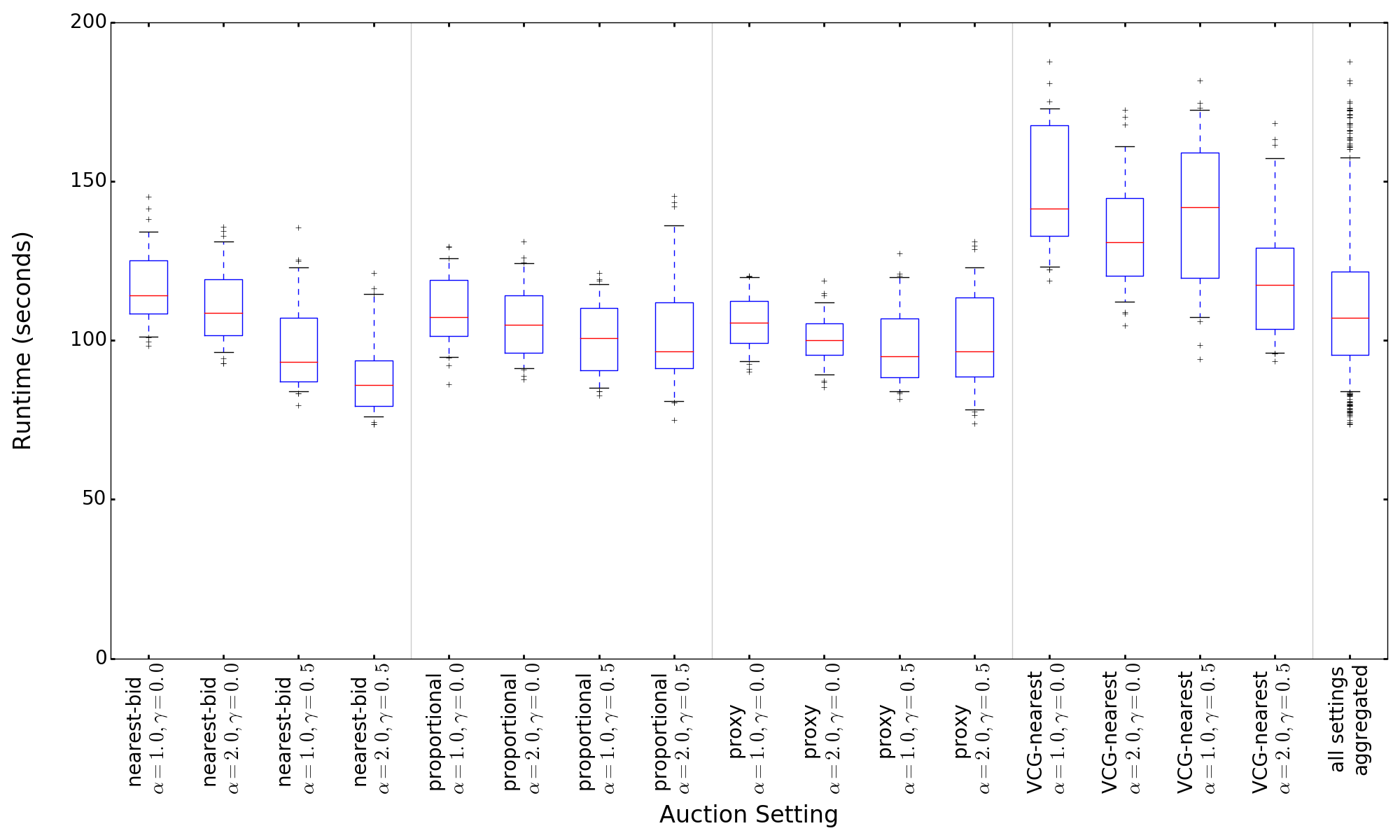}
\caption{Runtimes for the whole algorithm in \LLG including both the search and verification phases, for \numRuns runs on each of our 16 auction settings. All runs have been verified with $\epsilon < \targetEpsilon$. The  box plots provide the first, second and third quartiles; the whiskers are at the 5th and 95th percentile.}
\label{fig:finalruntimes}
\end{figure}

The runtimes of all \numSettingsTimesRuns runs are shown in Figure~\ref{fig:finalruntimes}.
Across all rules and settings, the median runtime of our whole algorithm is $107.1$ seconds (see the right-most box plot in the figure). Of course, there is some variance in the runtimes, both across our 16 settings and within each setting. However, overall, the runtimes are remarkably consistent, with 90\% of all auction instances completing within 83.9 and 157.5 seconds. Thus, a user of our BNE algorithm can expect similar runtimes when applying our  algorithm to similar settings. 

\subsection{Convergence Analysis}
\label{sec:compAnalyticResults}

\begin{table}
\newcommand{\numtable}[1]{\customlabeloutput{\num[round-precision=4]{#1}}}
\newcommand{\numlabel}[2]{\customlabel{#1}{\num[round-precision=4]{#2}}\numtable{#2}}
\begin{tabular}{l||r|r||c}
\multicolumn{1}{c||}{\textbf{Mechanism}} & \multicolumn{2}{c||}{\textbf{Standard Deviation of} $\epsilon$} & \multicolumn{1}{c}{\textbf{Max $L_\infty$ to}} \\
\multicolumn{1}{c||}{} &
\multicolumn{1}{c|}{\textbf{Search}} &
\multicolumn{1}{c||}{\textbf{Verification}} &
\multicolumn{1}{c}{\textbf{Analytical BNE}} \\
\hhline{=::=:=::=}
Nearest-Bid, $\alpha=1.0$, $\gamma=0.0$  & \numtable{4.272830e-07}           & \numtable{2.336476e-07} & \numtable{ 0.002749}        \\
Nearest-Bid, $\alpha=1.0$, $\gamma=0.5$  & \numtable{1.236745e-07}           & \numtable{6.342637e-08} & \numtable{ 0.001433}        \\
Nearest-Bid, $\alpha=2.0$, $\gamma=0.0$  & \numtable{4.401851e-07}           & \numtable{1.876482e-07} & \numtable{ 0.002450}        \\
Nearest-Bid, $\alpha=2.0$, $\gamma=0.5$  & \numtable{1.230580e-07}           & \numtable{6.661414e-08} & \numtable{ 0.001578}        \\
Proportional, $\alpha=1.0$, $\gamma=0.0$ & \numtable{1.245400e-09}           & \numtable{2.080031e-09} & \numtable{ 0.001442}        \\
Proportional, $\alpha=1.0$, $\gamma=0.5$ & \numtable{3.282607e-09}           & \numtable{2.819466e-09} & \numtable{ 0.000879}        \\
Proportional, $\alpha=2.0$, $\gamma=0.0$ & \numtable{1.658743e-09}           & \numtable{2.559018e-09} & \numtable{ 0.001536}        \\
Proportional, $\alpha=2.0$, $\gamma=0.5$ & \numtable{1.122485e-07}           & \numtable{3.906875e-09} & \numtable{ 0.001111}        \\
Proxy, $\alpha=1.0$, $\gamma=0.0$        & \numtable{1.414358e-09}           & \numtable{1.604076e-09} & \numtable{ 0.002520}        \\
Proxy, $\alpha=1.0$, $\gamma=0.5$        & \numlabel{StdevMax}{5.541368e-07} & \numtable{5.038464e-09} & \numtable{ 0.001557}        \\
Proxy, $\alpha=2.0$, $\gamma=0.0$        & \numtable{2.058928e-07}           & \numtable{1.561789e-09} & \numlabel{AnMax}{ 0.003865} \\
Proxy, $\alpha=2.0$, $\gamma=0.5$        & \numtable{5.306907e-09}           & \numtable{4.021574e-09} & \numtable{ 0.001728}        \\
Quadratic, $\alpha=1.0$, $\gamma=0.0$    & \numtable{1.285556e-09}           & \numtable{1.982474e-09} & \numtable{ 0.001442}        \\
Quadratic, $\alpha=1.0$, $\gamma=0.5$    & \numtable{2.618693e-09}           & \numtable{2.780691e-09} & \numtable{ 0.000881}        \\
Quadratic, $\alpha=2.0$, $\gamma=0.0$    & \numtable{1.854834e-09}           & \numtable{2.314903e-09} & \numtable{ 0.001522}        \\
Quadratic, $\alpha=2.0$, $\gamma=0.5$    & \numtable{1.595269e-07}           & \numtable{3.701131e-09} & \numtable{ 0.001111}        \\
\end{tabular}
\caption{Robustness measures of our algorithm in the LLG setting. The
  second and third columns show how much randomness  is introduced by Monte Carlo
  integration in the search and verification phase, respectively, measured as the
  standard deviation of $\epsilon$ over \numRuns runs each.  The
  forth column shows the maximum $L_\infty$ distance between the known analytical
  BNEs and our
  $\varepsilon$-BNEs ($\epsilon=\targetEpsilon$), also measured over \numRuns runs.}
\label{tab:compAnalytical}
\end{table}

When our algorithm has converged to a strategy profile $s^*$ and computed an $\epsilon$, how much certainty should we attribute to this result?
 Given that our algorithm uses random numbers quite extensively, one might worry that  $s^*$ and $\epsilon$ may be subject to large variance. To address this concern, %
we run some additional experiments and show that this randomness only negligibly affects the quality of our results.

However, testing the consistency of the $\epsilon$  is not  straightforward: the $\epsilon$ is ``self-reported'' by the algorithm, and thus any variance in $\epsilon$ could be due to differences in the strategy profiles found in the search phase, or differences in the computation of $\epsilon$ in the verification phase.
To deal with this, we consider the variance of the search and verification phases in isolation from each other, by running two separate experiments for each auction setting.
In the first experiment, we perform \numRuns runs with different random seeds during search, which are then run through verification using a single fixed seed.
This tells us the variance of $\epsilon$ caused by search.
In the second experiment, we again perform \numRuns runs, this time using a fixed random seed during search (thus converging on the exact same strategy profile), and different random seeds during verification.
This tells us the variance of $\epsilon$ caused by verification.

The results are shown in columns two and three of Table~\ref{tab:compAnalytical}.
We observe that the standard deviation of $\epsilon$ is extremely small in all cases, never exceeding \inref{StdevMax}.
This shows that, even though our algorithm is optimized to use as few Monte Carlo samples as possible, we can expect the algorithm's behavior to be very consistent  across different runs.

Now that we have shown that our BNE algorithm converges in a very consistent way, one may wonder about \emph{what} BNE strategy profile it converges to.
Recall that we have analytical results for all of our 16 settings. Thus, we can measure the distance between the true, analytical BNEs and the $\epsilon$-BNEs our algorithm finds. However, we would like to emphasize that minimizing this distance is \emph{not} part of the definition of an $\epsilon$-BNE and therefore not an objective of our algorithm. Nevertheless, some readers may be interested in knowing this distance. To this end, we consider the BNEs resulting from the previous experiment in each setting. In Table \ref{tab:compAnalytical}, we show the maximum (across all \numRuns runs) of the $L_\infty$ distances between the analytical BNEs from \citet{AusubelBaranov2013CoreOldVersion} and the $\targetEpsilon$-BNEs we find.\footnote{Originally, we found a discrepancy between some of the $\varepsilon$-BNEs reached by our algorithm (for Nearest-Bid with $\alpha=2$ and $0 \leq \gamma \leq 1)$  and the analytical BNEs in \citet{AusubelBaranov2013CoreOldVersion}. We contacted the authors, who promptly confirmed a small mistake in their analytical results and quickly provided us with the correct BNE formula for this setting, which is:
%
$s_i(v_i) = 1/\sqrt{8 \cdot (1 - \gamma)} \left( \log(\sqrt{2/(1-\gamma)} + v_i) - \log(\sqrt{2/(1-\gamma)} - v_i) \right)$. 
%
Our comparison in Table \ref{tab:compAnalytical} takes this correction into account.
Such a correction highlights the value of having numerical
techniques for finding BNEs, even in simple domains, as a complement
to analytical methods.}
As we can see, all  analytical BNEs are very close to our $\epsilon$-BNEs, never differing by more than \inref{AnMax}.
This also shows that our $\epsilon$-BNEs are close to each other, further confirming the low variance of our algorithm.

\section{Scaling to Higher Dimensions}
\label{sec:HigherDims}

In this section, we show how to scale our BNE algorithm to higher dimensions.
To this end, we introduce a new domain that is larger than LLG, called LLLLGG, where bidders are multi-minded,  i.e., they have a positive value for more than one bundle.
Numerically finding $\epsilon$-BNEs in this setting is much more complex compared to single-minded domains like  \LLG.
Thus, our algorithm significantly surpasses the previous state of the art, which allows us to analyze  new types of combinatorial interactions between bidders in much larger settings.

\subsection{The \MMLLLLGG Domain}

\begin{minipage}{\textwidth}
\begin{minipage}[b]{0.49\textwidth}
\centering
\begin{tabular}{c||c|c}
\textbf{Bidder} & \textbf{Bundle 1} & \textbf{Bundle 2} \\
\hhline{=#=|=}
$L_1$ & AB   & BC   \\
$L_2$ & CD   & DE   \\
$L_3$ & EF   & FG   \\
$L_4$ & GH   & HA   \\
$G_1$ & ABCD & EFGH \\
$G_2$ & CDEF & GHAB \\
\end{tabular}
\captionof{table}{Bundles of interest of each of the six bidders in the \MMLLLLGG domain with goods A-H. Each bidder is interested in exactly two bundles.}
\label{tab:MMLLLLGG}
\end{minipage}
\hfill
\begin{minipage}[b]{0.49\textwidth}
\centering
\includegraphics[width=0.85\textwidth]{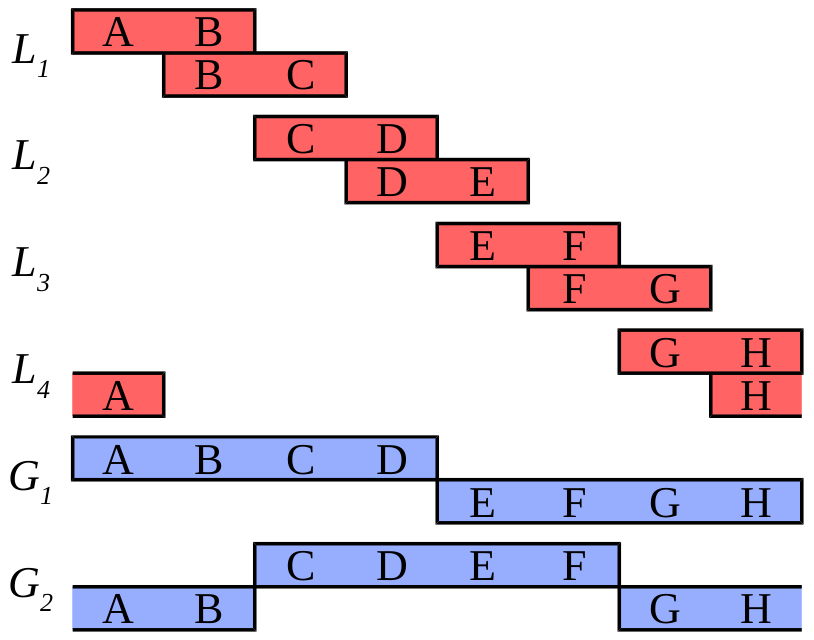}
\captionof{figure}{Graphical representation of bundle overlap in the \MMLLLLGG domain. \\}
\label{fig:MMLLLLGG}
\end{minipage}

\bigskip

\end{minipage}

We now introduce the \MMLLLLGG domain,
which is a synthetic domain along the lines of \LLG, but with significantly increased complexity.
The domain has six bidders and eight goods (labeled $A-H$ for ease of exposition), with each bidder being interested in two bundles as
enumerated in Table~\ref{tab:MMLLLLGG}.
There are four ``local'' bidders $L_1$ - $L_4$ who draw each of their two bundle values from $\Unif[0,1]$, and two ``global'' bidders $G_1$ and $G_2$ who draw their two bundle values from $\Unif[0,2]$; all draws are independent. 
For the global bidders, the two bundles are perfect substitutes (even though they are not overlapping),  such that a bidder will never want to win both bundles at the same time.
For the local bidders, their bundles of interest are partially overlapping such that they can obviously only ever win one of them. We assume straightforward bundle bidding as defined in Section~\ref{sec:formal_model}, and thus each bid $b_i$ consists of $r=2$ atomic bids.
Because the domain exhibits significant symmetries (see Figure \ref{fig:MMLLLLGG}), we can search for
symmetric equilibria where all local bidders have the same strategy and both
global bidders have the same strategy. Thus, a strategy profile is fully described by a pair of strategies $s_{local} : [0,1]^2
\mapsto \mathbb{R}_{\geq 0}^2$ and $s_{global} : [0,2]^2 \mapsto \mathbb{R}_{\geq 0}^2$.

\paragraph{The Source of Complexity.} Note that LLLLGG is larger than \LLG in several aspects: it has more bidders, more goods, each bidder bids on more bundles, and no bidder has a truthful strategy.
It is thus natural to ask how exactly these characteristics influence the difficulty of finding $\epsilon$-BNEs in this domain.
Some effects are very easy and intuitive to understand.
First and most importantly, to find a pointwise best response, a nonlinear optimization problem over the action space $\mathbb{R}_{\geq 0}^r$ must be solved, which is considerably harder when $r=2$ than when $r=1$.
Second, our algorithm's runtime increases linearly with the number of (non-symmetric) bidders, because best responses need to be computed for each of them.
Third, the more atomic bids we have, the harder the implementation of the mechanism itself becomes (e.g., because in general CAs, computing an efficient allocation or computing core payments are both NP-hard problems).
 
However, there is a source of complexity that is much subtler, namely the combinatorial way in which bidders interact.
To this end, we now want to provide some intuition. First, consider a first-price auction with a single good but \emph{many} bidders. In this auction, finding a best response $BR_i$ is not much more difficult than in the same auction with fewer bidders, because only the highest competing bid is relevant for bidder $i$'s utility, and other bids can be safely ignored.\footnote{This idea has recently been formalized by \citet{soumis2019return}, who define a variant of best response dynamics based on \emph{return functions}, which can explicitly model a large number of players by their aggregated effect on a game's outcome. Our algorithm framework is  compatible with this concept, as are most of the techniques we discuss in Section~\ref{sec:compBestResponses}.}
Similarly, having an auction with a large number of goods does not necessarily mean that finding an $\epsilon$-BNE is hard.
For instance, if two goods are perfect complements for all bidders, then this pair of goods will only ever be bid on together, effectively acting as a single good.
It is thus possible for an auction with many goods to have a combinatorial structure that is as simple as \LLG.

The true difficulty of an auction setting is determined by the way in which the bundles that bidders are interested in overlap with each other.
To make this concrete, consider the integral that must be solved each time the expected utility $\eu_i$ is computed.
It is well-known that the convergence rate of Monte Carlo integration is directly proportional to the variance of the function to be integrated, and is independent of the dimension \cite{press2007numerical}.
Thus, the number of Monte Carlo samples we need to ensure convergence only depends on how complex the distribution of bidder $i$'s utility is, not on how many bidders and goods are involved. The $\epsilon$-BNEs we compute next will demonstrate that the LLLLGG domain is indeed complex in this sense, exhibiting several novel and interesting equilibrium features.

\subsection{Experimental Results}

{\color{black}
For our experiments, we consider two auction mechanisms in the LLLLGG domain. Both use the efficient allocation rule, but they use different payment rules, namely \textit{VCG-nearest} and \textit{first-price}. The first-price payment rule is straightforward: It charges each bidder his bid on the bundle he wins.
VCG-nearest is more involved: It produces a payment vector that minimizes the Euclidean distance to the VCG payment vector, under the constraint that the payments are minimum-revenue core-selecting. Informally, this means that the payments have to be large enough to ensure that no coalition of bidders is willing to pay more in total than what the seller receives from the current winners. See \citet{DayRaghavan2007FairPayments} and \citet{DayCramton2012Quadratic} for formal definitions and further discussion of VCG-nearest.

We apply our full BNE algorithm to find $\epsilon$-BNEs for these two payment rules, using the verification method that produces a theoretical upper bound on $\epsilon$ (as described in Section~\ref{sec:provenErrorBound}).
%
The results are shown in Table~\ref{tab:MSDresults}. For each payment rule, we show two runs with different parameterizations. The runs reported in rows one and three target a highly accurate epsilon bound of $0.002$; this is the lowest $\varepsilon$ bound we could target with our algorithm using a reasonable number of core-hours.
The runs reported in rows two and four target a less accurate $\varepsilon$ bound of $0.01$; we include these runs to illustrate the runtime improvement that is possible if a less accurate $\varepsilon$ bound is acceptable.  For each run, we show the parameters that vary between the runs, i.e., the grid sizes and the number of Monte Carlo samples used in the search and verification phases. For all runs, as the pattern in pattern search, we use a grid of $3\times3$ points, with an initial grid spacing of $0.1$, and a budget of 8 steps in search and 12 steps in verification.}\footnote{Note that we do not use importance sampling and adaptive control points for the experiments described in this section, as these two techniques are not yet implemented in our code base for higher-dimensional value spaces. Of course, our algorithm would only get faster with these additions.}

\begin{table}[t]
\centering
\setlength\tabcolsep{2.3pt}
\newcommand{\numtable}[1]{\customlabeloutput{\num[round-precision=4]{#1}}}
\newcommand{\numtabletwo}[1]{\customlabeloutput{\num[round-precision=1]{#1}}}
\newcommand{\numlabel}[2]{\customlabel{#1}{\num[round-precision=4]{#2}}\numtable{#2}}

\begin{tabular}{l||c|c|c|r|r||r|r|r|c}
\bf Mechanism & \multicolumn{3}{c|}{\bf Grid Size} & \multicolumn{2}{c||}{\textbf{MC Samples}} & \multicolumn{3}{c|}{\bf Runtime (core-hours)} & \multirow{2}{1.9cm}{\bf \centering Upper Bound on $\epsilon$}  \\
 &  \bf Inner & \bf Outer  & \bf Verif. &   \textbf{Search}  &  \textbf{Verif}.  &  \bf Search  &  \bf Verif.  &  \bf Total  &   \\

\hhline{=::=:=:=:=:=::=:=:=:=}
VCG-nearest    &  $15\!\times\!15$  &  $25\!\times\!25$  &  $35\!\times\!35$ &  100,000  & 200,000  &  \numtabletwo{959.6981}  &  \numtabletwo{2960.9656}  &  \numtabletwo{3920.6638}  &   \numlabel{LLLLGGepsBoundQuadPrecise}{0.00185}  \\

VCG-nearest    &  $8\!\times\!8$  &  $12\!\times\!12$  &  $ 16\!\times\!16$  &  20,000  &  40,000  &  \numtabletwo{38.0268}  &  \numtabletwo{120.3453}  &  \numtabletwo{158.3721}  &   \numlabel{LLLLGGepsBoundQuadLoose}{0.009055}    \\

First-price    &  $50\!\times\!50$  &  $75\!\times\!75$  &  $100\!\times\!100$    &  100,000  & 200,000  &  \numtabletwo{82.5782}  &  \numtabletwo{115.19}  &  \numtabletwo{197.77}  &   \numlabel{LLLLGGepsBoundFPPrecise}{0.00190}    \\

First-price    &  $15\!\times\!15$  &  $20\!\times\!20$  &  $25\!\times\!25$ &  20,000  &  40,000  &  \numtabletwo{1.4024}  &  \numtabletwo{1.7885}  &  \numtabletwo{3.191}  &   \numlabel{LLLLGGepsBoundFPLoose}{0.009733}

\end{tabular}
\caption{Results of our BNE algorithm applied to VCG-nearest and first-price in the \MMLLLLGG domain (each row corresponds to one  run of the BNE algorithm). For each payment rule, we show two runs, one parameterized to achieve a small $\varepsilon$ bound and one parameterized to achieve a looser $\varepsilon$ bound.}\label{tab:MSDresults}
\end{table}

{\color{black}As the  outcome, we report the runtime in \textit{core-hours} and the obtained upper bound on $\varepsilon$.\footnote{Note that the runtime for VCG-nearest is orders of magnitude higher than that for first-price. This is to be expected, because the payment rule itself requires solving a quadratic program to find prices satisfying all core constraints, which is very expensive, even employing constraint generation \cite{DayRaghavan2007FairPayments,Buenz2015AFasterCCGAlgorithm}.  {\color{black} Also note that, in this experiment, we parallelized our BNE algorithm by performing the pointwise best response computations for all grid points in parallel on different cores of a large compute cluster.  However, we report the runtime in core-hours in a linearized way, i.e., as if the whole algorithm had been run sequentially on a single machine.}} As we see in rows one and three, our parameterizations for the highly accurate runs lead to an upper bound
on $\varepsilon$ of $0.0019$. Note that we used a  large number of MC samples for verification (i.e., \numSamplesMSDVerif) such that we have  high  confidence in the obtained $\epsilon$ bound.\footnote{\color{black}Based on our experience, we believe that the number of MC samples we used in verification is a conservative choice that produces a low-variance $\varepsilon$. In fact, our algorithm's performance in LLLLGG can likely be improved by using fewer MC samples. However, given that, in this domain, a detailed analysis of how the number of MC samples affects the variance of the $\varepsilon$ is quite costly, we leave this to future work. } No prior work has been able to compute an $\varepsilon$-BNE with this accuracy for a domain this complex.

As expected, we see that the runtime is significantly lower for the less accurate runs than for the more accurate runs.
This reduction in runtime is driven by multiple factors. First, when targeting a higher $\varepsilon$ bound, we can parameterize the search phase with coarser grid sizes and fewer MC samples. Second,  to verify an $\varepsilon$-BNE with a larger $\epsilon$ bound, a coarser grid in the verification phase suffices, and we also need fewer MC samples because a larger $\varepsilon$ bound can tolerate less precise pointwise best response computations. Given that the verification phase is the most expensive part of the BNE algorithm, the largest runtime gains come from the smaller grid size and the smaller number of MC samples in that phase.}

The final strategy profiles found for the highly accurate runs (with an $\varepsilon$ of 0.0019) for both VCG-nearest and first-price are depicted in Figures~\ref{fig:msdeqQuad} and~\ref{fig:msdeqFP}, respectively.
To enhance visual clarity, we show the equilibrium strategies in their piecewise bilinear form, before being converted to piecewise constant, as is required for our upper bound (see Section~\ref{sec:theoremApplication}).
We now discuss these strategies in more detail.

\begin{figure}[p]
\centering
\includegraphics[width=0.8\textwidth,trim={0 1.3in 0 0.7in},clip]{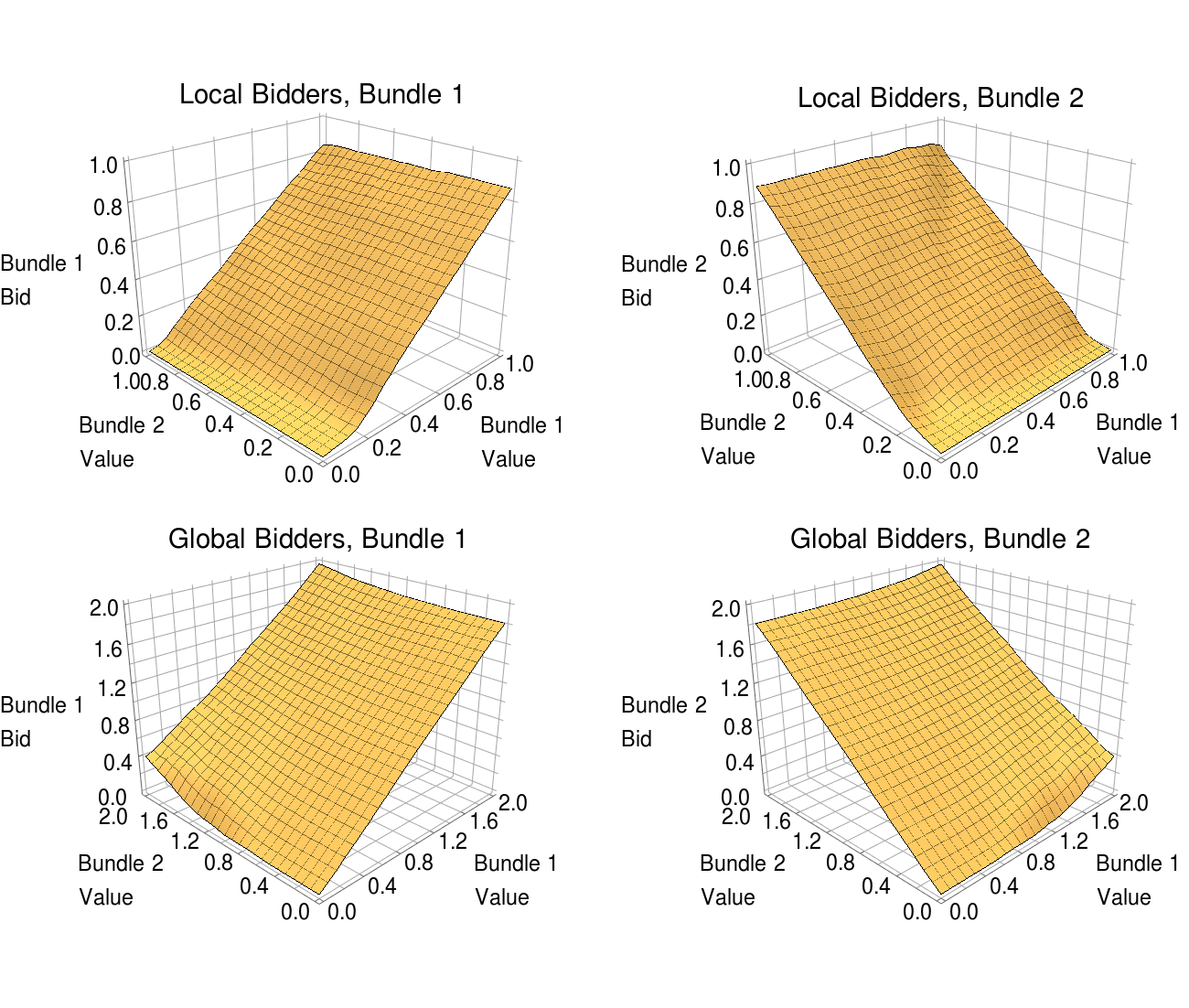}
\caption{ \inref{LLLLGGepsBoundQuadPrecise}-BNE of the \MMLLLLGG domain for the VCG-nearest payment rule. The top two plots show the
  BNE strategies for the local bidders for the two bundles they bid on, and the bottom two show the same for the global bidders.}
\label{fig:msdeqQuad}
\end{figure}

\begin{figure}[p]
\centering
\adjincludegraphics[width=0.8\textwidth,trim={0 1.3in 0 1.2in},clip]{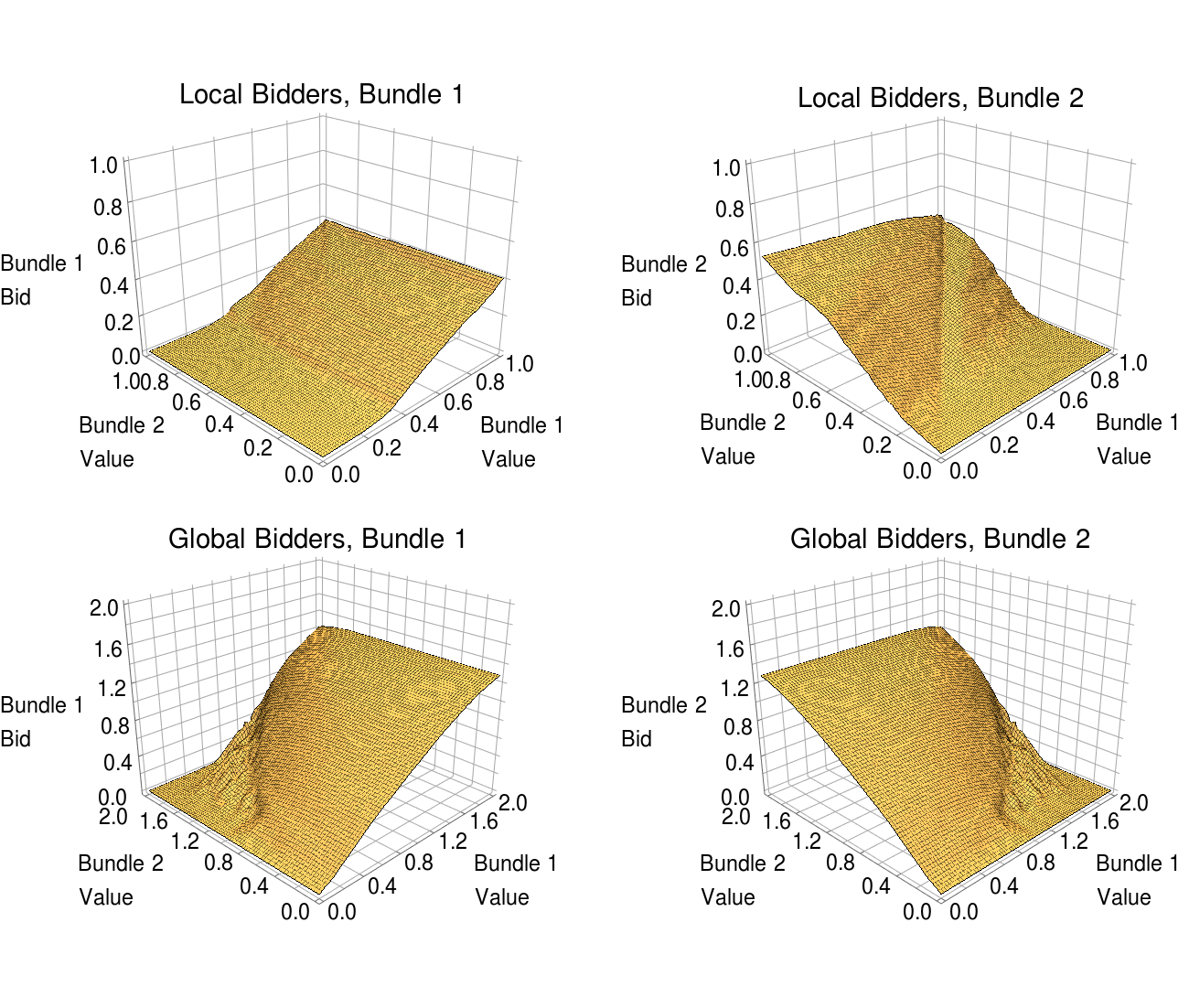}
\caption{\inref{LLLLGGepsBoundFPPrecise}-BNE of the \MMLLLLGG domain for the first-price payment rule. The top two plots show the
  BNE strategies for the local bidders for the two bundles they bid on, and the bottom two show the same for the global bidders.}
\label{fig:msdeqFP}
\end{figure}
\paragraph{Equilibrium of VCG-nearest.}

In Figure~\ref{fig:msdeqQuad}, we show the $\epsilon$-BNE for the VCG-nearest payment rule.
The first thing we observe in the equilibrium is that strategies are overall very close to truthful bidding. The most striking feature, however, is that the global bidders exhibit significant \emph{overbidding behavior}.
Specifically, the global bidders bid higher than their value on modestly valued bundles when their other bundle has a very high value.
At first sight, this may seem irrational. However, by overbidding the bidder shifts the VCG reference point, which can indirectly decrease his payment. For instance, suppose that bidders $G_1$, $L_3$ and $L_4$ win bundles $ABCD$, $EF$ and $GH$ respectively.
Bidder $G_1$ has an incentive to overstate the value for $EFGH$: this will cause the VCG reference point to increase for $L_3$ and $L_4$.
If the value for $ABCD$ is very high for $G_1$, then this manipulation carries little risk of accidentally winning the low-valued $EFGH$ bundle, and it decreases the expected payment for $ABCD$.
Of course, the bidder must balance the possible upside of overbidding with the potential risk of winning the lower-valued bundle at negative utility.
Thus, it is somewhat surprising to find such a strong overbidding effect in BNE. Note that the general concept of overbidding in CAs has first been identified by \citet{beck2013incentives}, who studied this phenomenon analytically in LLG for a stylized payment rule. We are the first to demonstrate that VCG-nearest, the rule most commonly used in practice, also exhibits this effect.

\paragraph{Equilibrium of first-price.}
In Figure~\ref{fig:msdeqFP}, we show the $\epsilon$-BNE for the first-price payment rule.
We observe that, in contrast to VCG-nearest, here the strategies involve significantly more shading. This is to be expected, given the incentive properties of the two rules.
{\color{black}
We also see that, for the local bidders, their strategy is to bid close to zero for one bundle in a large region of their value space (specifically, where their value for that bundle is low).
This behavior is known as free-riding. It allows a bidder to still win when other bidders make high bids on complementary bundles and obtain his bundle for low or even zero payment. If the bidder increased his own bid, the winning probability would not increase enough to justify the resulting higher payment.
While \citet{baranov2010exposure} has previously observed this effect in LLG (where each bidder only optimizes for a single allocation), it is interesting that it also shows up so clearly in the LLLLGG domain (where bidders' BNE strategies must take several possible allocations into account).}
For the global bidder, the  most striking feature of his equilibrium strategy is that, when the value for one bundle is high and the value for the other is low, the bid for the low-valued bundle is strongly shaded. This can again be explained by the structure of the LLLLGG domain, where the two bundles of interest of the global bidder are non-overlapping. Thus, in $\epsilon$-BNE, the global bidder's optimal strategy requires him to ``get out of his own way'' such that he avoids winning a low-profit bundle when instead he could be winning a much higher-profit bundle.

\bigskip
As we have shown in this section, our algorithm scales well to a domain that is significantly larger than LLG. We would like to emphasize that no prior analytical or algorithmic work has been able to solve domains of the size of LLLLGG.
We thus expect that other researchers will use this domain as a benchmark for future algorithmic work.

Our discussion of the VCG-nearest and first-price equilibria clearly demonstrates that optimizing the bidding strategy of a multi-minded bidder (as in LLLLGG) is significantly more complex than optimizing the strategy of a single-minded bidder (as in LLG). Note that we consider our experiments only as a first step towards exploring the properties of payment rules in larger domains where there are still many open questions that can be analyzed using our algorithm.

\section{Software Implementation of our BNE Algorithm}
\label{sec:software}

In this section, we describe the software implementation of our BNE algorithm, which is written in Java 8. 
In particular, we briefly explain how to use the software and then discuss its main features.
We release the software publicly under an open-source license at \texttt{https://github.com/} \texttt{marketdesignresearch/CA-BNE}. The code repository also contains three examples illustrating how to use our algorithm:
the VCG-nearest rule applied to \LLG, the first-price rule applied to \LLG (which is more challenging, since the global bidder is no longer truthful), and the first-price rule applied to LLLLGG (with a modest default configuration that can be run on a laptop).
%

\subsection{Usage}
\label{sec:Usage}

We provide example code in Listing~\ref{alg:javaContext} to show the overall structure of the code (we omit some details in the listing for the purpose of clarity).
The code is highly modular:
it consists of a series of different components, each responsible for one single aspect of the problem. Each component is implemented as a Java class, so one implementation can easily be swapped out for another. 
These components are plugged together via the \texttt{BNESolverContext} class. The overall algorithm is then managed and executed via  the \texttt{BNEAlgorithm} class.
We first create the context that holds the algorithm's configuration (Line 1).
Then we specify several of the algorithm's core elements, including pattern search\ (Line 3), Monte Carlo integration (Line 4), common random numbers (Line 5), and dampened updates (Line 6). 
Next, we specify the auction setting which consists of a 
mechanism class that specifies the allocation and payment rules (Line 8), and a sampler class that specifies the value distributions of bidders (Line 9).\footnote{Note that the sampler class provides a method for sampling from the marginal value distribution $\mathcal{V}_\smi | V_i = v_i$ for a given $v_i$. This design choice makes our Monte Carlo sampling more efficient, more so than if the valuations were implemented explicitly as a joint probability distribution.}
Then the BNE Algorithm class is initialized with the context carrying all the configuration data (Line 10).
Finally, the actual BNE algorithm is started (Line 12).

\RestyleAlgo{boxed}

\begin{listing}[tb]
\small

\texttt{BNESolverContext<Double,Double> context = new BNESolverContext<>()}\;
\texttt{\ldots}\\
\texttt{context.setOptimizer(new PatternSearch<>(...))}\;
\texttt{context.setIntegrator(new MCIntegrator<>(...))}\;
\texttt{context.setRng(2, new CommonRandomGenerator(...))}\;
\texttt{context.setUpdateRule(new UnivariateDampenedUpdateRule(...)}\;
\texttt{\ldots}\\
\texttt{context.setMechanism(new FirstPrice())}\;
\texttt{context.setSampler(new FirstPriceLLGSampler(context))}\;
\texttt{BNEAlgorithm<Double,Double> bneAlgo = new BNEAlgorithm<>(3,context)}\;
\texttt{\ldots}\\
\texttt{result = bneAlgo.run()}\;
\caption{Example code for running our BNE Algorithm. The ``...''s indicate code or function arguments that we omit for clarity.}
\label{alg:javaContext}
\end{listing}

\subsection{Main Features}

We wrote our code with performance, ease of use, and extensibility in mind. 
Note that these three goals are sometimes at odds with each other, but we have attempted to do justice to each of them, through careful software engineering.  Our code provides the following main features:

\begin{enumerate}
    \item \textbf{Flexible Bidder Models.} The value distributions and utility functions are simply implemented as classes. All possible joint value distributions are supported, even those with complex correlation structures. Nonlinear utilities (e.g., risk aversion) can easily be added. The action space is not restricted to  $\mathbb{R}_{\geq 0}^r$, as our algorithm can also handle domain-specific bidding languages, for example. 
    \item \textbf{Flexible Mechanism Specification.} It is straightforward to implement new allocation and payment rules because they are simply provided by the user as new  classes. This also allows for writing fast, specialized implementations (e.g., for domains where special-purpose algorithms for the winner determination problem are available) or the addition of features like reserve prices.
    \item \textbf{No hard-coded Assumptions.} Our BNE algorithm can also be run in settings where the assumptions we make in Section~\ref{sec:provenErrorBound} to derive the theoretical upper bound on $\epsilon$ are not satisfied. In these settings, our algorithm simply falls  back to our alternative verification method that computes an estimated $\epsilon$.
\end{enumerate}

With these features, our code is  easy to use and extend. It can thus immediately be used in exploratory research and as a basis for further algorithmic work. We hope that providing access to our source code will enable the study of mechanisms and domains that were not previously amenable to analytic or algorithmic analysis.

\section{Conclusion}
\label{sec:conclusion}

In this paper, we have presented a new algorithm for finding pure-strategy $\epsilon$-BNEs in CAs with continuous value and action spaces. The key difference between our approach and prior work is that our algorithm is separated into a search phase and a verification phase. This has enabled us to design a search phase that is highly optimized for speed while still obtaining a provable bound on $\epsilon$ in the verification phase.

To address the high dimensionality that drives the difficulty of the BNE search problem, sampling-based approaches are needed, which in turn introduces variance. Therefore, we have employed multiple variance reduction techniques in the search phase. Additionally, the key to obtain large speed improvements was to think carefully about when and where to exert computational effort (e.g., using adaptive control points). In the theory section, our key assumption was  the linearity of the utility function. This enabled us to  make the necessary convexity argument to characterize the gap between the equilibrium utility and the best response utility and to ultimately derive our theoretical bound.
Surprisingly, we were able to derive this bound without making any assumptions about the allocation or payment rules, while at first sight, one may have expected that certain assumptions (like Lipschitz-continuity) would have been necessary. For practical applications, it is important to note that our algorithm computes an upper bound that is relatively "tight," in the sense that, with sufficiently many verification points, the upper bound is very close to the true $\epsilon$.

Our BNE algorithm can be used in several ways. First, researchers can use it, e.g., for the purpose of auction design, for analyzing certain aspects of CAs, or for validating analytical results. Importantly, by enabling the algorithmic analysis of multi-minded domains like LLLLGG,  researchers can now explore many new effects that are not present in simple domains like LLG. Second, bidders can use our algorithm, e.g., to analyze the effect of different strategies or when contemplating whether to enter an auction in the first place. We hope that releasing our code under an open-source license will enable those and many other applications in the future.

Finally, our work gives rise to promising directions for future research. 
First, an important question is how to derive a bound for settings with correlation or non-linear utilities (which includes risk-averse bidders). Note that without assumptions on the mechanism, no bound for correlated settings can be derived. Thus, future work could explore which assumptions on the mechanism are necessary/sufficient to extend our theoretical bound to richer settings.
Second, in the long run, it would be interesting to develop BNE algorithms that scale to real-world-sized problems. As we have discussed, the main difficulty in scaling our algorithm lies in the richness of the bidding space. Thus, future work targeting scalability could investigate how to exploit certain structure inherent to the problem of interest (e.g., monotonicity properties of the mechanism) to focus the computational effort on the most relevant parts of the value or action spaces.

\acks{We thank the editor and the anonymous reviewers for their very helpful comments. Furthermore, we thank Enrique Areyan, Amy Greenwald, Daniel Marszalec, and Alex Teytelboym for trying early versions of our BNE algorithm code and for providing useful feedback.}


\vskip 0.2in
\bibliography{../../../CApapers}
\bibliographystyle{theapa}
\newboolean{withAppendix}
\setboolean{withAppendix}{true}

\ifthenelse{\boolean{withAppendix}}{
\newpage
\appendix


\section{Pseudocode for Adaptive Control Point Placement}
\label{app:adaptiveCP}

\SetKwFunction{pointwiseBestResponse}{PointwiseBestResponse}
\SetKwFunction{computePriority}{CurvatureBasedPriority}
\SetKwFunction{makePWL}{MakePiecewiseLinearStrategy}

Here, we describe the details of our adaptive control point placement algorithm as introduced in Section~\ref{sec:adaptiveCP}. We only describe it for the one-dimensional case (as needed for LLG) and it would have to be generalized to work in multiple dimensions. Recall that the method is based on the idea of repeatedly finding a valuation where the best response function has high curvature, such that placing an additional control point there improves how well the corresponding piecewise linear strategy approximates the best response.

The pseudocode for our adaptive control point placement method is provided in Algorithm~\ref{alg:adaptiveCP}. The input to the algorithm is a function   \pointwiseBestResponse which, for a fixed bidder $i$, takes a control point $v_i$ as its only argument and computes the bid $b_i$ that maximizes the expected utility $\eu_i(v_i, b_i)$.
This function thus encompasses all the details about the auction mechanism $\mathcal{M}$, the joint value distribution $\mathcal{V}$, and the current strategy profile $s$.

\begin{algorithm}[]
\SetAlgoLined
\DontPrintSemicolon
\SetKwInOut{KwInput}{input}
\SetKwInOut{KwOutput}{output}
\SetKwInOut{KwParameters}{parameters}

\KwInput{function \pointwiseBestResponse that computes a pointwise best response for a fixed bidder $i$, at any valuation $v_i$}
\KwParameters{initial number of control points $initialPoints$\\ maximum number of control points $maxPoints$}
\KwOutput{strategy $s_i'$}

$controlPoints$ :=  list of valuations initialized to an evenly spaced grid of size $initialPoints$\;
$bestResponses :=$ map of (valuation, pointwise best response) pairs, initially empty \;
$priorities :=$  map of (valuation, priority) pairs, initially empty \;

\textbf{foreach} $v_i \in controlPoints$: compute $b_i:=$ \pointwiseBestResponse{$v_{i}$}  and add $(v_i, b_i)$  to $bestResponses$\;
\textbf{foreach} $v_i \in controlPoints$: compute $\pi_i:=$ \computePriority{$v_{i},bestResponses$} and add $(v_i,\pi_i)$ to $priorities$ \;

\Repeat{\upshape $|controlPoints| = maxPoints$}{
    $v_i :=$ control point with highest priority in $priorities$\;
    $v_i' :=$  neighbor of $v_i$ in $controlPoints$ furthest away from $v_i$ (and in case of a tie, select the neighbor $v_i'$ where the  difference between the best responses at $v_i$ and $v_i'$ is larger)\;
$v_i'' := (v_i + v_i') / 2$\;

    add $v_i''$ to $controlPoints$ in between $v_i$ and $v_i'$ \;
    compute $b_i'':=$ \pointwiseBestResponse{$v_i''$}  and add $(v_i'',b_i'')$ to $bestResponses$\;
    \textbf{foreach} $w_{i} \in \{v_i, v_i',v_i''\}$: compute $\pi_i:=$ \computePriority{$w_{i},bestResponses$} and add/update  $(w_{i},\pi_i)$ to/in $priorities$\;

}

$s_i' :=$ \makePWL{$bestResponses$} \;

\Return $s_i'$ \;

\caption{Computing a best response for one bidder with adaptive control point placement}
\label{alg:adaptiveCP}
\end{algorithm}

Algorithm~\ref{alg:adaptiveCP} works as follows.
In Lines 1-3, we initialize the data structures that keep track of all valuations we have selected as control points, their associated pointwise best responses, and a priority score for each of them.
The initial set of control points is a coarse, evenly spaced grid.
In Line 4, we compute the pointwise best responses for the initial control points. In Line 5, we then compute a priority for each of those control points. For this, we make use of the helper function  \computePriority which assigns a priority score for a given control point based on the curvature of the best response function at that control point (see Algorithm~\ref{alg:priority}).
In each iteration of the main loop of the algorithm (Lines 6-13), we  add one new control point.
For this, we first find the control point $v_i$ with the highest priority  and then place a new control point at the midpoint between $v_i$ and one of its neighbors (Lines 7-10).
Concretely, the new control point is placed on the side where the distance to the neighboring control point is larger;
in case of a tie, we select the side where the slope of the best response is estimated to be steeper.
We then compute a pointwise best response at the new control point and recompute all priority scores that might have changed (Lines~11-12).
When we reach the maximum number of allowed control points (Line 13), we construct a piecewise linear strategy based on all the control points and their pointwise best responses (Line~14). Finally, we return the new strategy (Line 15).

\bigskip
\begin{algorithm}[H]
\SetAlgoLined
\DontPrintSemicolon
\SetKwInOut{KwInput}{input}
\SetKwInOut{KwOutput}{output}
\SetKwInOut{KwParameters}{parameters}
\SetKwInOut{KwFunction}{function}
\SetKwProg{mySub}{function}{}{}

\KwFunction{\computePriority{$v_i$, $bestResponses$}}
\KwInput{control point $v_i$, map of (valuation, bid) pairs $bestResponses$}
\KwParameters{smallest permissible interval size $minIntervalSize$}
\KwOutput{priority $\pi_i \in \mathbb{R}_{\geq0}$}

    $\underaccent{\bar}{v}_i, \bar{v}_i := $ lower and upper neighbors of $v_i$ in $bestResponses$\;
    $b_i, \underaccent{\bar}{b}_i, \bar{b}_i :=$ bids corresponding to $v_i$, $\underaccent{\bar}{v}_i$, and $\bar{v}_i$ in $bestResponses$, respectively \;
    
    \If{$\min \left( v_i - \underaccent{\bar}{v}_i, \bar{v}_i - v_i \right) \leq minIntervalSize$}{
        \Return 0\;
    }

    $lowerSlope := (b_i - \underaccent{\bar}{b}_i) / (v_i - \underaccent{\bar}{v}_i) $ \;
    $upperSlope := (\bar{b}_i - b_i) / (\bar{v}_i - v_i) $ \;
    \Return $|upperSlope - lowerSlope|$ \;


\caption{Computing the priority of a control point based on the curvature of a bidder's strategy}
\label{alg:priority}
\end{algorithm}
\bigskip

To complete the description of our method we provide pseudocode for the helper function \computePriority in Algorithm~\ref{alg:priority}. The inputs to the algorithm are the control point $v_i$  for which we compute a priority and a map of best responses for the set of current control points.\\ In Lines 1-2, we find the two neighboring control points of $v_i$ and retrieve the best responses associated with $v_i$ and its two neighbors.
If either of $v_i$'s neighbors is too close to $v_i$ then we return 0 (the lowest possible priority) to ensure that the grid is not subdivided too many times in any one region (Lines 3-5).
Otherwise, we approximate the second derivative of the best response at $v_i$ through finite differences and return the magnitude of this derivative as the priority of $v_i$ (Lines 6-8).


}{} 

\end{document}